\documentclass[11pt]{article}

\usepackage{tikz}
\usepackage{float}
\usepackage{mathtools}
\usepackage{scalerel}
\usepackage{amssymb}
\usepackage{amsfonts}
\usepackage{amsthm}
\usepackage{amsmath}
\usepackage{dsfont}
\usepackage{mathrsfs}
\usepackage{enumitem}
\usepackage[T1]{fontenc}
\usepackage{chet}
\usepackage{slashed}
\usepackage{multirow}
\usepackage[font=small,format=hang,labelfont={sf,bf}]{caption}
\usepackage{subcaption}
\usepackage{caption}
\usepackage{placeins}
\usepackage{breqn}
\usepackage{pgfplots}

\usepgfplotslibrary{colorbrewer}
\pgfplotsset{width=10cm,height=10cm,compat=1.9}

\newtheorem*{theorem}{Theorem}

\usetikzlibrary{intersections,arrows.meta}
\usetikzlibrary{patterns}

\definecolor{darkblue}{rgb}{0.1,0.1,0.7}
\hypersetup{colorlinks,
           linkcolor={darkblue},
           citecolor={darkblue},
           urlcolor={darkblue},
           pdftitle={A note on defect stability in d=4-ε},
           pdfdisplaydoctitle
}

\tikzset{cross/.style={path picture={
      \draw[black]
            (path picture bounding box.south east) --
            (path picture bounding box.north west)
            (path picture bounding box.south west) --
            (path picture bounding box.north east);}}}

\newcommand{\lsp}{\hspace{0.5pt}}

\title{A note on defect stability in $d=4-\varepsilon$}

\author{William H.\ Pannell\emails{(\href{mailto:william.pannell@kcl.ac.uk}{william.pannell})}}

\affiliation{Department of Mathematics, King's College London, Strand, London WC2R 2LS, United Kingdom}

\abstract{We explore the space of scalar line, surface and interface defect field theories in $d=4-\varepsilon$ by examining their stability properties under generic deformations. Examples are known of multiple stable line defect Conformal Field Theories (dCFTs) existing simultaneously, unlike the case of normal multiscalar field theories where a theorem by Michel guarantees that the stable fixed point is the unique global minimum of a so-called $A$-function. We prove that a suitable modification of Michel's theorem survives for line defect theories, with fixed points locally rather than globally minimizing an $A$-function along a specified surface in coupling space and provide a novel classification of the fixed points in the hypertetrahedral line defect model. For surface defects Michel's theorem survives almost untouched, and we explore bulk models for which the symmetry preserving defect is the unique stable point. For interface defects we prove only the weaker condition that there exist no fixed points stable against generic deformations for $N\geq 6$.
}

\date{August 2024}

\begin{document}

\maketitle

\toc

\section{Introduction}
Wilson's renormalization group (RG) methods\cite{Wilson:1971bg,Wilson:1971dh} provide powerful tools for investigating the characteristics of a wide range of critical phenomena. In particular, the investigation of RG fixed points in $d=4-\varepsilon$ provides strong estimates for the critical exponents of three-dimensional statistical models at second-order phase transitions\cite{Wilson:1971dc}. Multiscalar models, with $N$ massless scalar fields $\phi_i$, $i=1,\ldots, N$, interacting via the Lagrangian
\begin{equation}
    \frac{1}{2}\partial_\mu \phi_i\partial^\mu\phi_i+\frac{\lambda_{ijkl}}{4!}\phi_i\phi_j\phi_k\phi_l\,,
\end{equation}
have been well-studied in this regard\cite{Osborn:2017ucf,Osborn:2020cnf}. In the Landau-Ginzburg paradigm the fixed points of this Lagrangian contain information about the phase transitions in many interesting physical models\cite{Pelissetto:2000ek}, such as the liquid-gas interface\cite{Hubbard,Brilliantov}, magnetic anisotropies in ferromagnetic materials with a hypercubic crystalline structure\cite{Aharony:1973zz}, and the n-component Potts spin model\cite{Zia:1975ha}. These fixed points satisfy
\begin{equation}
    \beta_{ijkl}=-\varepsilon\lambda_{ijkl}+\left(\lambda_{ijmn}\lambda_{mnkl}+\text{Perms.}\right)=0\,,
    \label{eq:betabulk}
\end{equation}
where $\beta_{ijkl}=d\lambda_{ijkl}/d\ln{\mu}$ governs the renormalization of the interaction couplings $\lambda_{ijkl}$ as a function of the energy scale $\mu$. It was proven by Michel\cite{Michel:1983in} that stable fixed points of (\ref{eq:betabulk}) globally minimize a particular function, $A$, and are thus unique whenever they exist. This $A$ is constructed by noting that the beta function is gradient, i.e. it can be written in the form
\begin{equation}
    \beta^{I}=\frac{\partial A(\lambda)}{\partial \lambda_{I}}\,,
    \label{eq:bulkgraddef}
\end{equation}
where $I=(ijkl)$ is a generalized index. This expression severely constrains the behavior of renormalization flows in multiscalar models, as such an $A$ must monotonically decrease along flows towards the IR
\begin{equation}
    \frac{d A}{d\ln{\mu}}=\partial_{I}A\beta^{I}=\partial_{I}A\partial^{I}A\geq0\,,
    \label{eq:monotonicity}
\end{equation}
and it is this fact which permits Michel's theorem. More generically, gradient beta functions may be written in the form
\begin{equation}
\beta^I=T^{IJ}(\lambda)\partial_J A(\lambda)\,,
\label{eq:gradflowdef}
\end{equation}
where $T^{IJ}$ is a symmetric metric on coupling space assumed to be positive definite. The inclusion of a metric in this definition will not alter the monotonicity properties of $A$ as long as it is positive definite, as (\ref{eq:monotonicity}) is replaced by
\begin{equation}
    \frac{d A}{d\ln{\mu}}=\partial_{I}A\beta^{I}=\partial_{I}A\lsp T^{IJ}\lsp \partial_{J}A\geq0\,.
\end{equation}
It has been shown by Wallace and Zia that at three-loops and above the inclusion of a non-flat metric is necessary\cite{Wallace:1974dy}, and recent results indicate that the multiscalar beta function is gradient through six-loops provided some conditions are met\cite{Jack:2016tpp,Pannell:2024sia}.

Recently, there has been growing interest in studying the perturbation of bulk models by the introduction of defect operators, describing localized impurities in otherwise critical systems\cite{Cuomo:2021kfm,Giombi:2022vnz,Pannell:2023pwz,Drukker:2022pxk,CarrenoBolla:2023sos,CarrenoBolla:2023vrv,Barrat:2023ivo,Trepanier:2023tvb,Sakkas:2024dvm,Bianchi:2023gkk,Aharony:2023amq,Hu:2023ghk,Harribey:2023xyv,Harribey:2024gjn,Ge:2024hei,Lauria:2020emq}. These defect operators can be characterized by the dimension of their support $p<d$ where $d$ is the dimension of the ambient space. In this paper we will be interested in defects inserted in a dimension $d=4-\varepsilon$ critical bulk scalar model, where the introduction of a defect perturbs the action by a weakly relevant operator and triggers an RG flow to a non-trivial defect CFT (dCFT) fixed point. To reach a given dCFT one must set to zero any relevant deformations that exist at that fixed point, corresponding in a physical system to tuning a number of parameters to arrive at a phase transition in the dCFT's universality class. Totally stable fixed points, which have zero relevant operators, are special in this regard, and can be reached without any tuning at all. It is natural in the study of the space of dCFTs to ask whether or not Michel's theorem continues to hold for these defect models. That is to say, are stable dCFTs unique and classified by the global minimization of some function $A$? As noted in \cite{Pannell:2023pwz} the answer to this question is negative for the line defect; there may exist multiple stable line defect fixed points derived from from the same critical bulk theory. Examples in that paper were explicitly constructed for $N=7$ and $N=9$ scalars, with the bulk taken to lie at the hypertetrahedral fixed point with $S_{N+1}\times \mathbb{Z}_2$ symmetry. 

The meaning of the word stability depends upon the context in which it is used. Very often for both practical and physical purposes one is not interested in considering all of coupling space, and instead wishes to restrict to a small submanifold consisting of deformations preserving a certain group $G$ of field transformations. Crucially, Michel's theorem  continues to hold for all such submanifolds in multiscalar theories. More generic deformations are useful in the study of crossover effects corresponding to the introduction of an anisotropy into critical systems. Many anisotropies in bulk multiscalar models, such as a cubic anisotropy in an $O(N)$ system, are of experimental and general interest. In this paper we will be concerned with deriving results for generic deformations, some of which may also be restricted to apply to stability preserving deformations. For clarity we will use totally stable to refer to conclusions which are only applicable to stability with respect to generic deformations.

We seek to expand knowledge about defect fixed points by 
investigating Michel's theorem for line defects ($p=1$), surface defects ($p=2$), and co-dimension 1 interface defects. This paper is organized in the following manner: In section \ref{michelsec} we review Michel's theorem for multiscalar models in order to apply the method of proof to the various defect models. In section \ref{linesec} we indicate how the proof of Michel's theorem fails for line defect models. We then demonstrate that the beta function can be made gradient by an unusual choice of $A$-function and metric and use this fact to prove that a local version of Michel's theorem survives in these systems. We verify our theorem by re-deriving the analysis of the fixed points in the $O(N)$ and hypercubic models presented in \cite{Pannell:2023pwz}, and then go further to provide a novel classification of fixed points of the hypertetrahedral model. In section \ref{surfacesec} we consider surface defects, proving that stable fixed points globally minimize a suitably defined $A$-function and providing an explicit analysis of stability for a number of bulk models. Finally, in section \ref{interfacesec} we investigate interface defects, proving that for $N\geq 6$ there exist no totally stable fixed points, and arguing that for $N<6$ the only possible totally stable fixed point is the trivial interface.

\section{Michel's Theorem for multiscalar models}\label{michelsec}
For the convenience of the reader, let us first recall how Michel's theorem works in the case of bulk multiscalar models. 
\begin{theorem}
    If $\lambda^*_{ijkl}$ is a stable fixed point of (\ref{eq:betabulk}), then it must correspond to the unique global minimum of the function
    $$A=-\frac{\varepsilon}{2}\lambda_{ijkl}\lambda_{ijkl}+\lambda_{ijkl}\lambda_{klmn}\lambda_{ijmn}$$
    in the space of fixed points.
\end{theorem}
\begin{proof}We follow the proof by Rychkov and Stergiou\cite{Rychkov:2018vya}. As noted by Wallace and Zia\cite{Wallace:1974dx,Wallace:1974dy}, the one-loop beta function (\ref{eq:betabulk}) can be written as the total derivative of the so-called $A$-function,
\begin{equation}
    \beta_{ijkl}=\frac{\partial A}{\partial \lambda_{ijkl}}\,,\qquad A=-\frac{\varepsilon}{2}\lambda_{ijkl}\lambda_{ijkl}+\lambda_{ijkl}\lambda_{klmn}\lambda_{ijmn}\,.
    \label{eq:scalargrad}
\end{equation}
One can easily see that this $A$ will be a Liapunov function\cite{liapunov2016stability}, monotonically decreasing along RG flows towards the IR, as
\begin{equation}
    \frac{d A}{d \ln{\mu}}=\frac{\partial A}{\partial \lambda_{ijkl}}\frac{d \lambda_{ijkl}}{d\ln{\mu}}=\beta_{ijkl}\beta_{ijkl}\geq 0\,
\end{equation}
with $dA/d \ln{\mu}=0$ only at fixed points satisfying $\beta_{ijkl}=0$. At the fixed point, we have that
\begin{equation}
    \varepsilon\lambda^*_{ijkl}=\lambda^*_{ijmn}\lambda^*_{mnkl}+\text{Perms.}\,,
    \label{eq:fixedpointsolbulk}
\end{equation}
so that the critical values of $A$ will be given by
\begin{equation}
    A(\lambda^*)=-\frac{\varepsilon}{6}\lambda^*_{ijkl}\lambda^*_{ijkl}\leq 0\,.
    \label{eq:Afixedpointbulk}
\end{equation}
Linearizing the dynamical system $d\lambda_{ijkl}/d\ln{\mu}=\beta_{ijkl}$ about a fixed point, one sees that the stability of fixed points is determined by the eigenvalues of the stability matrix
\begin{equation}
    S_{ijkl,mnop}=\frac{\partial \beta_{ijkl}}{\partial \lambda_{mnop}}\,,
\end{equation}
which from (\ref{eq:scalargrad}) is equal to the Hessian of $A$. Suppose that we have two distinct fixed points, $\lambda^1_{ijkl}\neq\lambda^2_{ijkl}$, and let us restrict ourselves to the $\lambda^1-\lambda^2$ plane in coupling space. We can choose coordinates $(s,t)$ such that $\lambda^1$ lies at the origin like
\begin{equation}
    \lambda_{ijkl}=(1+s)\lambda^1_{ijkl}+t\lambda^2_{ijkl}\,.
\end{equation}
For simplicity, let us introduce the notation of \cite{Rychkov:2018vya} and write
\begin{equation}
    \lambda_{ijkl}\lambda_{ijkl}=(\lambda,\lambda)\,,\qquad \lambda_{ijkl}\lambda_{klmn}\lambda_{ijmn}=(\lambda,\lambda,\lambda)\,.
\end{equation}
Then, $A(\lambda)$ in the $(s,t)$ coordinates is given by
\begin{equation}
    A(\lambda)=(\lambda_1,\lambda_1)\left(-\frac{1}{6}+\frac{s^2}{2}+\frac{s^3}{3}\right)+(\lambda_2,\lambda_2)\left(-\frac{t^2}{2}+\frac{t^3}{3}\right)+(\lambda_1,\lambda_2)(1+s)(s+t)t\,,
\end{equation}
where we have used (\ref{eq:fixedpointsolbulk}) to write everything in terms of $(\lambda,\lambda)$, and have set $\varepsilon=1$ to remove an overall factor for notational simplicity. The Hessian of $A$ about $\lambda_1$ is thus given by the matrix
\begin{equation}
    S=\begin{pmatrix}
        (\lambda_1,\lambda_1) && (\lambda_1,\lambda_2) \\
        (\lambda_1,\lambda_2) && 2(\lambda_1,\lambda_2)-(\lambda_2,\lambda_2)
    \end{pmatrix}\,,
\end{equation}
which has eigenvalues
\begin{equation}
    \frac{1}{2}\left(a_1+2b-a_2\pm \sqrt{(a_1-2b)^2+(a_2-2b)^2+2a_1 a_2}\right)\,,
    \label{eq:bulkeigs}
\end{equation}
where we have defined
\begin{equation}
    a_1=(\lambda_1,\lambda_1)\,, \quad a_2=(\lambda_2,\lambda_2)\,, \quad b=(\lambda_1,\lambda_2)\,.
\end{equation}
As $a_1,a_2\geq0$, one sees that taking the plus sign in (\ref{eq:bulkeigs}) will always lead to a positive eigenvalue,
\begin{equation}
    a_1+2b-a_2+\sqrt{(a_1-2b)^2+(a_2-2b)^2+2a_1 a_2}\geq a_1+2b-a_2+|2b-a_2|\geq0\,,
\end{equation}
so that the stability of $\lambda_1$ is determined by the sign of the second eigenvalue. Reorganizing the terms in the square root, one sees that this is
\begin{equation}
    a_1+2b-a_2-\sqrt{( a_1+2b-a_2)^2+4(a_1-b)^2+4a(a_2-a_1)}\,.
\end{equation}
If $a_2>a_1$, this is manifestly negative. If $a_2=a_1$ this will also be negative unless 
$a_1=b$, but this leads to the contradiction
\begin{equation}
    (\lambda_1-\lambda_2,\lambda_1-\lambda_2)=a_1+a_2-2b=0\,,
\end{equation}
which violates the assumption that $\lambda^1_{ijkl}\neq \lambda^2_{ijkl}$. Thus, if $a_2\geq a_1$, $\lambda_1$ will be unstable with respect to perturbations towards $\lambda_2$. Note that the minus sign in (\ref{eq:Afixedpointbulk}) translates this condition to $A(\lambda_2)\leq A(\lambda_1)$. The theorem then follows.
\end{proof}

\section{Stability in scalar line defects}\label{linesec}
Let us first consider the case of a dimension $p=1$ line defect. The effect of this defect can be represented via the addition of a deformation to the action localized on the defect
\begin{equation}
    S=S_{\text{bulk}}+S_{\text{defect}}=\int d^dx\left(\frac{1}{2}\partial_\mu \phi_i\partial^\mu\phi_i+\frac{\lambda_{ijkl}}{4!}\phi_i\phi_j\phi_k\phi_l\right)+\int_{-\infty}^{\infty} d\tau h_i\phi_i(\tau,\Vec{0})\,,
\end{equation}
where $\tau$ is a proper time along the defect and where we take the bulk interaction $\lambda_{ijkl}$ to lie at a fixed point of (\ref{eq:betabulk}). In $d=4-\varepsilon$ the couplings $h_i$ will be weakly relevant with a classical dimension of $\Delta_h=\varepsilon/2$, so that $S_{\text{defect}}$ will trigger an RG flow in the space of defect couplings. To one loop, the beta function for this RG flow has been calculated in \cite{Cuomo:2021kfm,Giombi:2022vnz,Pannell:2023pwz} and is given by
\begin{equation}
    \beta_i=-\frac{\varepsilon}{2}h_i+\frac{1}{6}\lambda_{ijkl}h_j h_k h_l\,.
    \label{eq:betadefect}
\end{equation}

Now, let us try to follow the logic of Michel's theorem in the case of line defects, where we will find that there will be an obstruction preventing us from reaching the same conclusion. Examining (\ref{eq:betadefect}), one sees that we can again write the beta function as the derivative of an $A$-function, where now
\begin{equation}
    A(h)=-\frac{\varepsilon}{4}h_ih_i+\frac{1}{24}\lambda_{ijkl}h_ih_jh_kh_l\,.
\end{equation}
While strictly speaking vector indices $i$ are to be raised and lowered using the flat metric $\delta^{ij}$, we will follow the convention of always keeping lowercase Latin indices lowered. Fixed points $h^*$ will satisfy the equation
\begin{equation}
    3h^{*}_i=\lambda_{ijkl}h^*_j h^*_k h^*_l\,,
\end{equation}
so that at the fixed point $A$ will take the value
\begin{equation}
    A(h^*)=-\frac{\varepsilon}{8}h^*_ih^{*}_i\leq0\,.
\end{equation}
Again let us consider two non-identical fixed points $h^1_i\neq h^2_i$ and restrict ourselves to the $h^1-h^2$ plane in coupling space. Parameterizing this plane by
\begin{equation}
    h_i=(1+s)h_i^1+th^2_i\,
\end{equation}
the $A$-function becomes (again setting $\varepsilon=1$)
\begin{equation}
\begin{split}
    A(h)=&(h_1,h_1)\left(-\frac{1}{8} + \frac{s^2}{2} + \frac{s^3}{2} + \frac{s^4}{8}\right)+(h_2,h_2)\left(-\frac{t^2}{4}+\frac{t^4}{8}\right)\\ &+(h_1,h_2)(1 + s)(2 s + s^2 + t^2)\frac{t}{2}+(h_1,h_1,h_2,h_2)_\lambda\frac{(1+s)^2t^2}{4} \,,
\end{split}
\end{equation}
where we have introduced the notation
\begin{equation}
    (h_1,h_1,h_2,h_2)_\lambda=h^1_ih^1_j\lambda_{ijkl}h^2_kh^2_l\,.
\end{equation}
Crucially, $(h_1,h_1,h_2,h_2)_\lambda$ has no definite sign, which will prove to be an obstruction in the proof. The Hessian of $A$ now becomes
\begin{equation}
    S=\begin{pmatrix}
        (h_1,h_1) && (h_1,h_2) \\
        (h_1,h_2) && \frac{1}{2}(h_1,h_1,h_2,h_2)_\lambda-\frac{1}{2}(h_2,h_2)
    \end{pmatrix}\,,
    \label{eq:hessiandefect}
\end{equation}
with eigenvalues
\begin{equation}
    \frac{1}{2}\left(a+c\pm\sqrt{a^2+4b^2-2ac+c^2}\right)\,,
\end{equation}
where $a=(h_1,h_1)$, $b=(h_1,h_2)$ and $c=(h_1,h_1,h_2,h_2)_\lambda/2-(h_2,h_2)/2$. As before, one can see that one of these eigenvalues is necessarily positive, as choosing a plus sign yields
\begin{equation}
    a+c\pm\sqrt{a^2+4b^2-2ac+c^2}\geq a+c+|a-c|\geq0\,,
\end{equation}
where the last equality follows as $(h_1,h_1)\geq0$. The other eigenvalue takes the form
\begin{equation}
    a+c-\sqrt{a^2+4b^2-2ac+c^2}=a+c-\sqrt{(a+c)^2+4b^2-4ac}\,.
\end{equation}
If $4b^2-4ac\geq0$ then this will indeed be negative and $\lambda_1$ will be unstable, but this condition relies on knowledge about the tensor $\lambda_{ijkl}$, whose precise form is unknown. Indeed, as may be expected, this condition is violated in the examples with multiple stable fixed points. For the $N=7$ hypertetrahedral model, for instance, one finds that
\begin{equation}
    4b^2-4ac=-\frac{6075}{128}
\end{equation}
for the two stable fixed points, indicating that there exists some sort of barrier preventing an RG flow from connecting the two points. The nature of this barrier we will now make more precise.

\subsection{A local criterion of stability}\label{stab1sec}
The obstruction to Michel's theorem in the case of line defects arose from the inclusion in $A$ of the term
\begin{equation}
    A(h)\supset\frac{1}{6}\lambda_{ijkl}h_ih_jh_kh_l\,.
\end{equation}
Considering (\ref{eq:gradflowdef}), we notice that we can remove this term from $A$ while retaining its important monotonicity properties by introducing corrections to the flat metric. Specifically, we choose to absorb the $\lambda h^3$ term in the beta function entirely within the metric as
\begin{equation}
    T_{ij}=\delta_{ij}-\frac{1}{3\varepsilon}\lambda_{ijkl}h_kh_l\,, \qquad A(h)=-\frac{\varepsilon}{4}h_ih_i\,.
    \label{eq:Tdef}
\end{equation}
Naively, one might assume that this choice of $A$ completely obstructs Michel's theorem, as the Hessian (\ref{eq:hessiandefect}) becomes
\begin{equation}
    S=\begin{pmatrix}
        -\frac{1}{2}(h_1,h_1) && -\frac{1}{2}(h_1,h_2) \\
        -\frac{1}{2}(h_1,h_2) && -\frac{1}{2}(h_2,h_2)
    \end{pmatrix}\,,
\end{equation}
which has a negative determinant by the Cauchy-Schwarz inequality. However, the inclusion of a non-trivial metric means that the stability matrix of the system is now related to the Hessian of $A$ by
\begin{equation}
    S_{ij}=T_{ik}\partial_{k}\partial_j A\,,
\end{equation}
so that negative eigenvalues of the Hessian no longer guarantee negative eigenvalues of $S^i_j$.

There is a salient feature of this solution which demands note, namely the explicit non-perturbative factor of $1/\varepsilon$ which appears in $T_{ij}$. As we have fixed the bulk to lie at a critical point, one must remember that $\lambda_{ijkl}$ will be $O(\varepsilon)$, so that the pole in $\varepsilon$ will cancel and both terms in $T_{ij}$ will be $O(\varepsilon^0)$. This cancellation means that the one loop correction to the flat metric will not be suppressed, allowing it to compete and, depending upon the form of the interaction tensor, potentially violate the assumption of positive definiteness. Usually, gradient flow is been explored with an eye towards generalizing Zamolodchikov's $c$-theorem\cite{Zamolodchikov:1986gt} to provide of RG-monotonicity theorems for a variety of physical theories. For instance Cardy's $A$-theorem\cite{Cardy:1988cwa} was proven perturbatively by Jack and Osborn who proved the existence of a solution to (\ref{eq:gradflowdef}) in $d=4$\cite{Jack:1990eb}. In that context the function $A$ is expected to encode physically meaningful information about the states in the QFT, being related to the the $a$-anomaly of the trace of the stress tensor. Here, $A=-\varepsilon h_ih_i/4$ has no such physical interpretation, and is purely a mathematical artifice that will prove to be useful when examining the stability of fixed points. As we are only interested in setups containing non-zero fixed points, we will not seek to take any $\varepsilon\rightarrow 0$ limit and factors of $1/\varepsilon$ will remain finite.

For small $h$ $T_{ij}$ will approximate the flat metric and thus satisfy positive-definite condition, so that $A$ must monotonically decrease along RG flows close to the origin. This can be seen explicitly by considering the evolution of $A$ under RG flow:
\begin{equation}
    \frac{d A}{d\ln{\mu}}=\partial_i A T_{ij} \partial_j A=\frac{\varepsilon^2}{4}(r^2-\frac{1}{3\varepsilon}\lambda_{ijkl}h_ih_jh_kh_l)\,,
    \label{eq:Aflow}
\end{equation}
where here $r^2=h_ih_i$ is the radius in coupling space. Coupling space is thus divided into regions: a region $\mathcal{M}$ around the origin within which $d A/d\ln{\mu}>0$ and a number of regions where $d A/d\ln{\mu}<0$, the number of which depends upon the form of $\lambda_{ijkl}$. These regions are separated by the surfaces defined by the equation $h_iT_{ij}h_j=0$, along which $d A/d\ln{\mu}=0$. These surfaces can equivalently be defined as those along which $\beta_i$ has no radial component in spherical coordinates. Crucially, $h^*_i\beta_i(h^*)=0$ at fixed points, so all non-trivial fixed points must lie on these surfaces. It is the definition of these surfaces that will allow us to give a sense of 'locality' to Michel's theorem. The main result for line defects is then the following theorem:
\begin{theorem}
    Non-trivial fixed points of (\ref{eq:betadefect}) locally extremize $A=-\varepsilon h_ih_i/4$ on the surface(s) defined by $h_i T_{ij}h_j=0$, with local minima corresponding to stable fixed points.
\end{theorem}
\begin{proof}
For ease of visualization let us consider the equivalent problem of extremizing the radius $r^2$, the maximization of which corresponds to the minimization of $A$. Flows within $\mathcal{M}$ containing the trivial fixed point will move to outwards towards the boundary (or possibly off to infinity if the surfaces defined by $h_i T_{ij}h_j=0$ are non-compact), and we must concern ourselves with the behavior of the flows close to the boundary. To make the definition of the boundary precise, let us first change to spherical coordinates $(r,\Vec{\Omega})$, and solve (\ref{eq:Aflow}). This equation becomes
\begin{equation}
    \frac{\varepsilon}{4}r^2(1-\frac{1}{3\varepsilon}\lambda_{ijkl}\hat{h}_i\hat{h}_j\hat{h}_k\hat{h}_l r^2)
\end{equation}
where the $\hat{h}_i$ are vectors on the the unit sphere $S^{N-1}$ and depend only on the angular coordinates $\Vec{\Omega}$. This can be solved to give the boundary $r(\Vec{\Omega})$ as
\begin{equation}
    r(\Vec{\Omega})=\sqrt{\frac{3\varepsilon}{\lambda_{rrrr}}}\,,
    \label{eq:rsol}
\end{equation}
where $\lambda_{rrrr}=\lambda_{ijkl}\hat{h}_i\hat{h}_j\hat{h}_k\hat{h}_l$ is the purely radial component of the bulk interaction tensor. The components of the beta function in these new coordinates are best expressed using $T^{ij}$ as
\begin{equation}
    \beta_a=-\frac{\varepsilon}{2}T_{ar}r.
\end{equation}
From (\ref{eq:Tdef}) one sees that $\beta_r$ simply reproduces (\ref{eq:Aflow}) and will thus vanish along the boundary as expected. On the boundary the angular components of the beta function will be
\begin{equation}
    \beta^\alpha=\frac{1}{6}\frac{g^{\alpha\beta}}{r^2}\lambda_{\beta rrr}r^3=\frac{g^{\alpha\beta}\lambda_{\beta rrr}}{2\lambda_{rrrr} }\,,
\end{equation}
where $g^{\alpha\beta}$ is the inverse metric on the unit sphere $S^{N-1}$. Here we briefly break with our index convention by raising the index on $\beta$ to remind the reader that it transforms as a vector. To simplify this expression, we notice that
\begin{equation}
    \lambda_{\beta rrr}=r\frac{\partial\hat{h}_i}{\partial\Omega^\beta}\hat{h}_j\hat{h}_k\hat{h}_l\lambda_{ijkl}=r\frac{1}{4}\partial_\beta(\lambda_{rrrr})=-\frac{\lambda_{rrrr}}{2}\partial_\beta r(\vec{\Omega})\,.
\end{equation}
Thus, along the boundary we have
\begin{equation}
    \beta^{\alpha}=-\frac{\varepsilon}{4}\nabla^\alpha r(\Vec{\Omega})\,,
    \label{eq:betaboundary}
\end{equation}
where $\nabla^\alpha r$ is the gradient of $r$ taken as a function on $S^{N-1}$. Note that in the last line we have absorbed a factor of $r$ by rescaling the angular vectors to live on the unit sphere. One sees immediately that fixed points $\beta^\alpha=0$ correspond to points extremizing $r(\Omega)$.

The crucial feature of this expression is that the gradient of a function points in the direction of steepest ascent. As flows towards the IR are governed by $-\beta_i$ rather than $\beta_i$, this means that the vector field along the boundary will point towards points on the surface maximizing $r(\Vec{\Omega})$, but without any radial component. These vectors thus point inward towards $\mathcal{M}$ rather than tangent to the surface, so that flows can only move from outside $\mathcal{M}$ to inside, and not vice versa. One fact remains necessary to prove the theorem. As one can easily check, at all non-trivial fixed points $h_i^*$ will be an eigenvector of the stability matrix
\begin{equation}
    S_{ij}=\partial_j\beta_i=-\frac{\varepsilon}{2}\delta_{ij}+\frac{1}{2}\lambda_{ijkl}h_kh_l
\end{equation}
with eigenvalue $\varepsilon$, so that all fixed points are stable against radial perturbations. We thus get the following picture: At fixed points corresponding to local minima or saddle points of $r(\vec{\Omega})$, at least one of the non-radial perturbations will lie inside of $\mathcal{M}$, and will thus trigger a flow which increases in radius. As the boundary is repulsive to internal flows, these flows cannot re-approach the original fixed point by crossing into the external region, and thus correspond to unstable directions. At fixed points corresponding to local maxima, all of the non-radial perturbations will lie in the external region, and thus trigger flows which decrease in radius. However, by (\ref{eq:betaboundary}) and the continuity of the beta function, these flows must point back towards the fixed point and thus correspond to stable directions. One sees that stable fixed points are in one-to-one correspondence with points on the surface(s) locally maximizing $r(\Vec{\Omega})$ or equivalently locally minimizing $A$, and the theorem is proven.
\end{proof}

It is important to notice that this theorem is also applicable to symmetry preserving flows. Restriction to the hyperplane corresponding to deformations preserving a particular symmetry group $G$ is equivalent to considering the gradient to be restricted to the symmetry preserving hypersphere inside $S^{N-1}$. Fixed points which are stable under symmetry preserving deformations must minimize the radius only with respect to these tangent directions, though of course fixed points which are totally stable must continue to be stable.

This theorem is inherently perturbative, and relied in part on the fact that there was only a single term at one loop in the quantum beta function. Nevertheless, it is expected that the stability properties of a fixed point at one loop will hold to all loop orders, as the higher loop corrections to the eigenvalues of the stability matrix will be suppressed by powers of $\varepsilon$. For small enough $\varepsilon$, higher order terms will not alter the behavior of $d A/d\ln{\mu}$ in the vicinity of $r(\Vec{\Omega})$ beyond small perturbations in (\ref{eq:rsol}). That is to say the behavior of the flows in the vicinity of the surface should be qualitatively unchanged by the inclusion of higher loop orders in $\beta_i$, with $A$ monotonically decreasing inside the surface and monotonically increasing just outside. Thus, stability or instability informed by the extremization of the one loop $r(\vec{\Omega})$ as in the theorem will be sufficient to guarantee stability or instability perturbatively.

To help visualize this theorem in action, and to demonstrate its power to correctly account for multiple stable fixed points, let us turn now to some examples.

\subsection{\texorpdfstring{$O(N)$}{O(N)} model}
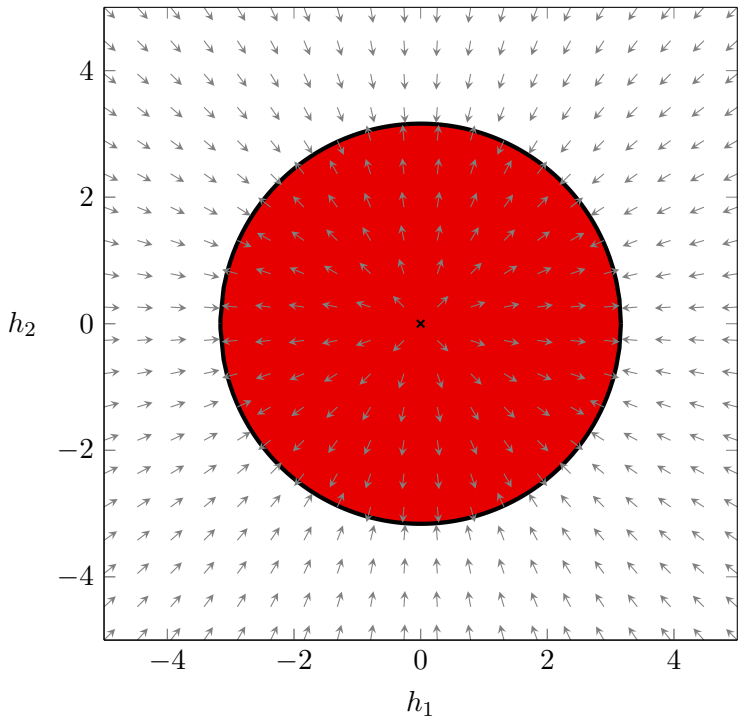
\begin{figure}[t!]
    \centering
\begin{tikzpicture}
\begin{axis}[
    xmin=-5, xmax=5,
    ymin=-5, ymax=5,
    xlabel= $h_1$,
    ylabel= $h_2$,
    ylabel style={rotate=-90},
    xticklabel style={scaled ticks=false, /pgf/number format/fixed, /pgf/number format/precision=3},
    yticklabel style={scaled ticks=false, /pgf/number format/fixed, /pgf/number format/precision=3},
    title={Solution validity for a line defect in an $O(2)$ bulk},
    legend pos = outer north east,
    view={0}{90},
    domain=-5:5
]
\draw[fill=red!90!black,opacity=0.5] (axis cs:0,0) circle [black,radius=sqrt(10)];
\addplot+[only marks, mark=x,mark options={color=black,fill=black}, thick] coordinates{(0,0)};
\addplot+[
    no marks,color=black,ultra thick]
table{datafiles/O2minus.dat};
\addplot+[
    no marks,color=black,ultra thick]
table{datafiles/O2plus.dat};
\addplot3 [
        gray,-stealth,samples=20,
        quiver={
            u={-(-(x/2) + x*(x^2 + y^2)/20)/sqrt((-(x/2) + x*(x^2 + y^2)/20)^2+(-(y/2) + y*(x^2 + y^2)/20)^2)},
            v={-(-(y/2) + y*(x^2 + y^2)/20)/sqrt((-(x/2) + x*(x^2 + y^2)/20)^2+(-(y/2) + y*(x^2 + y^2)/20)^2)},
            scale arrows=0.25,
        },
    ] (x,y,0);
\end{axis}
\end{tikzpicture}
    \caption{Perturbative defect RG flows within an $O(2)$ bulk. The red region indicates the region $\mathcal{M}$. There is a single fixed point modulo the action of $O(2)$, which lives along the circle defined by $r^2=N+8$.}
    \label{fig:O2metric}
\end{figure}
The simplest case one can consider is that of a scalar line defect placed inside of an $O(N)$ critical bulk, where the interaction tensor takes the form
\begin{equation}
    \lambda_{ijkl}=\frac{\varepsilon}{N+8}\left(\delta_{ij}\delta_{kl}+\delta_{ik}\delta_{kl}+\delta_{il}\delta_{jk}\right)\,,
\end{equation}
For $N=2$, the beta function (\ref{eq:betadefect}) depicted in Figure \ref{fig:O2metric}. Using (\ref{eq:Tdef}), the metric through one loop is
\begin{equation}
    T_{ij}=\delta_{ij}\left(1-\frac{h_kh_k}{3N+24}\right)-h_ih_j\frac{2}{3N+24}\,,
\end{equation}
which has two distinct eigenvalues for $N>1$
\begin{equation}
    1-\frac{h_ih_i}{N+8}\,,\qquad 1-\frac{h_ih_i}{3N+24}\,,
\end{equation}
where the first eigenvalue corresponds to the eigenvector $h_i$. One sees that coupling space is divided in two by the $N-1$ sphere $r^2=N+8$, with 
\begin{equation}
    \frac{d A}{d\ln{\mu}}=T_{ij}\lsp\partial_i A\lsp\partial_j A\begin{cases} >0\,, & h^2<N+8\\
    <0\,, & h^2>N+8\\
\end{cases}
\end{equation}
which agrees with the existence of a family of IR stable fixed points lying on the sphere\cite{Pannell:2023pwz}. This can also be seen by considering the boundary $r(\vec{\Omega})$ as defined above. Here,
\begin{equation}
    \lambda_{rrrr}=\frac{3\varepsilon}{N+8}
\end{equation}
is a constant function on $S^{N-1}$, so that every point on the sphere $r^2=N+8$ corresponds to a stable fixed point. As the beta function is zero along the sphere, there will be no flows between the fixed points themselves, and all RG flows must begin at either $h_i=0$ or at infinity. It is important to mention that we must properly define fixed points to live within coupling space modulo the action of the bulk symmetry group. The $O(N)$ bulk symmetry will act transitively on the sphere of fixed points, so that flows are defined only along a single ray, and the stable fixed point is in fact still unique here.

\subsection{Hypercubic model}
\begin{figure}[t!]
    \centering
\begin{tikzpicture}
\begin{axis}[
    xmin=-5, xmax=5,
    ymin=-5, ymax=5,
    xlabel= $h_1$,
    ylabel= $h_2$,
    ylabel style={rotate=-90},
    xticklabel style={scaled ticks=false, /pgf/number format/fixed, /pgf/number format/precision=3},
    yticklabel style={scaled ticks=false, /pgf/number format/fixed, /pgf/number format/precision=3},
    title={Solution validity for a line defect in an $D_4=\mathbb{Z}_2^2\ltimes S_2$ bulk},
    legend pos = outer north east,
    view={0}{90},
    domain=-5:5
]
\addplot+[
    no marks,color=black,fill=red!90!black,opacity=0.5]
table{datafiles/B2.dat};
\addplot+[only marks,mark=*,mark options={color=black,fill=black}, ultra thick] coordinates{(2.12132,2.12132) (-2.12132,2.12132) (2.12132,-2.12132) (-2.12132,-2.12132)};
\addplot+[only marks, mark=square*,mark options={color=black,fill=black}, ultra thick] coordinates{(4.24264,0) (0,4.24264) (-4.24264,0) (0,-4.24264)};
\addplot+[only marks, mark=x,mark options={color=black,fill=black}, thick] coordinates{(0,0)};
\addplot3 [
        gray,-stealth,samples=20,
        quiver={
            u={-(-(x/2) - x^3/18 + x*(x^2 + y^2)/12)/sqrt((-(x/2) - x^3/18 + x*(x^2 + y^2)/12)^2+(-(y/2) - y^3/18 + y*(x^2 + y^2)/12)^2)},
            v={-(-(y/2) - y^3/18 + y*(x^2 + y^2)/12)/sqrt((-(x/2) - x^3/18 + x*(x^2 + y^2)/12)^2+(-(y/2) - y^3/18 + y*(x^2 + y^2)/12)^2)},
            scale arrows=0.25,
        },
    ] (x,y,0);
\end{axis}
\end{tikzpicture}
    \caption{Perturbative defect RG flows within a $D_4=\mathbb{Z}_2^2\ltimes S_2$ bulk. The red region indicates the region $\mathcal{M}$. Beyond this region $A$ will monotonically increase. There are two non-trivial fixed points modulo $B_2$ transformations, indicated by squares and circles. The unique stable fixed point being found by minimizing the value of $A$ along the boundary.}
    \label{fig:B2metric}
\end{figure}
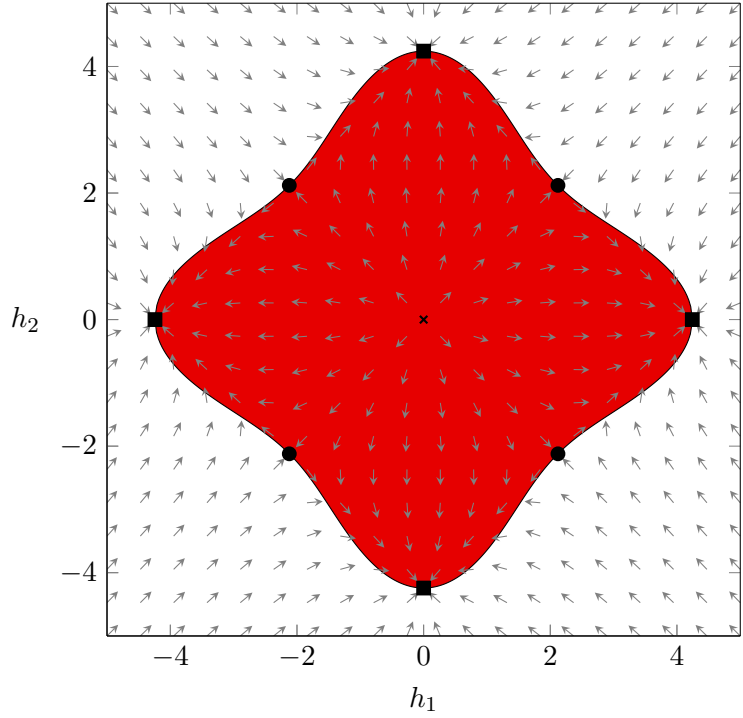
For a hypercubic scalar model, with symmetry $B_N=\mathbb{Z}_2^N\ltimes S_N$, the bulk interaction tensor will be
\begin{equation}
    \lambda_{ijkl}=\frac{\varepsilon}{3N}\left(\delta_{ij}\delta_{kl}+\delta_{ik}\delta_{kl}+\delta_{il}\delta_{jk}\right)+\frac{(N-4)\varepsilon}{3N}\delta_{ijkl}\,.
    \label{eq:hypercubicfp}
\end{equation}
The one loop metric in this case is then given by
\begin{equation}
    T_{ij}=\delta_{ij}\left(1-\frac{h_kh_k}{N}-\frac{(N-4)h_i^2}{9N}\right)-h_ih_j\frac{2}{9N}\,.
\end{equation}
Once again, $T_{ij}\sim \delta_{ij}$ for small $h_i$, so that $A$ will monotonically decrease close to the origin. The surface bounding this region is determined by
\begin{equation}
    r(\vec{\Omega})=\sqrt{\frac{9N}{3+(N-4)\sum_i\hat{h}_i^4}}\,,
    \label{eq:rhypercubic}
\end{equation}
where the $\hat{h}_i$ are to be expressed in spherical coordinates. Fixed points are thus classified by the manner in which they extremize $f(\Vec{\Omega})=\sum_i\hat{h}_i^4$ on the unit $N-1$ sphere. To determine the form of the fixed points, let us consider the function $f$ evaluated at an arbitrary unit vector $\hat{h}$. Varying the value of one of the coordinates, $\hat{h}_j$, will alter the corresponding term in $f$ by $\delta \hat{h}_j^4=\hat{h}_j^3\delta h_j$. The restriction to $S^{N-1}$ means that any change in $\hat{h}_j$ must also be accompanied by a change in the other coordinates as well, given by the constraint
\begin{equation}
    \sum_{i\neq j}\hat{h}_i\delta h_i=-\hat{h}_j \delta h_j\,,
\end{equation}
with different choices for the $\delta h_i$, $i\neq j$ corresponding to movement along different curves on the sphere. Choosing coordinates such that a single $\delta h_i$, $i\neq j$, is non-zero, one sees that
\begin{equation}
    \delta f=\hat{h}_j^3\delta h_j+\hat{h}_i^3\delta h_i=(\hat{h}_j^2-\hat{h}_i^2)\hat{h}_j\delta h_j\,.
\end{equation}
We thus immediately find that $\hat{h}$ is only an extremal point if all of its non-zero components have the same magnitude. The vector on the unit sphere $\hat{h}$ does not give a fixed point of the beta function on its own, and we must rescale by a factor of $r(\hat{h})$ to construct the true fixed points. If $1\leq m\leq N$ of the components are non-zero, we can use the $B_N$ symmetry to put $\hat{h}$ into the form
\begin{equation}
    \hat{h}_i=\begin{cases} \frac{1}{\sqrt{m}}\,, & i\leq m\\
    0\,, & i>m\\
\end{cases}\,,
\end{equation}
so that we have
\begin{equation}
    r(\hat{h})=\sqrt{\frac{9N}{3+\frac{(N-4)}{m}}}\,.
\end{equation}
The fixed point is then given by
\begin{equation}
        \hat{h}_i=\begin{cases} \sqrt{\frac{9N}{3m+(N-4)}}\,, & i\leq m\\
    0\,, & i>m\\
\end{cases}\,,
\end{equation}
which matches the form of the fixed points found in \cite{Pannell:2023pwz}.

The prefactor of $(N-4)$ is crucial in this analysis, as the sign will determine whether stable fixed points will maximize or minimize this function. For $N<4$ this prefactor will be negative, so that stable fixed points will correspond to maxima of $f$\footnote{While one might worry that these maxima would result in a denominator that is negative, and thus not correspond to fixed points, we take $N\in \mathbb{N}$ here. The only troublesome case here is $N=1$, where $\hat{h}_i^4=1$ and the denominator vanishes. However, the bulk interaction tensor (\ref{eq:hypercubicfp}) also vanishes for $N=1$, so that this simply reflects the lack of any defect fixed points in a free bulk.}. As one can easily convince oneself, these maxima occur precisely along the coordinate axes. As in the $O(N)$ model, these points lie in a single orbit under $B_N$ symmetry, which acts by permuting and reflecting the axes, and thus the stable fixed point can again be considered unique in this case. For $N=4$, $r(\Vec{\Omega})=12$ is constant on $S^{N-1}$, and the case is the same as $O(4)$ considered above. For $N>4$, radial maxima will correspond to minima of $f(\Vec{\Omega})$, that is points on $S^{N-1}$ of maximal distance from each of the coordinate axes. Considering the coordinate axes to lie at the vertices of a hypercube, these minima will lie on the $2^N$ faces, and as before lie in a single orbit under $B_N$. The stable fixed point is thus unique for all $N$. This analysis is identical to that found in \cite{Pannell:2023pwz}, which looked explicitly at the stability matrix and its eigenvalues.

For $N=2$ the situation is shown in Figure \ref{fig:B2metric}. Here, the bulk model is equivalent to two decoupled Ising models, which can be seen with the field redefinition
\begin{equation}
    \phi_1\rightarrow \phi'_1=\frac{\phi_1+\phi_2}{\sqrt{2}}\,,\quad     \phi_2\rightarrow \phi'_2=\frac{\phi_1-\phi_2}{\sqrt{2}}\,,
\end{equation}
which rotates the axes in the figure by $\pi/4$. The one-way nature of flows at boundary is made clear by the figure, which also demonstrates nicely how flows within $\mathcal{M}$ pool at points maximizing $r(\vec{\Omega})$. For $N=2$ (\ref{eq:rhypercubic}) takes the explicit form on $S^1$
\begin{equation}
    r(\theta)=\frac{6}{\sqrt{3-\cos{4\theta}}}\,.
    \label{eq:B2boundary}
\end{equation}
Along the boundary the beta function takes the form
\begin{equation}
    \vec{\beta}=\frac{3\sin{4\theta}}{(3-\cos{4\theta})^{3/2}}\lsp\varepsilon\,\hat{\theta}\,,
    \label{eq:B2boundarybeta}
\end{equation}
where $\hat{\theta}$ is the unit angular vector field. Examining (\ref{eq:B2boundary}) and (\ref{eq:B2boundarybeta}), one finds the equality
\begin{equation}
    \vec{\beta}=-\frac{\varepsilon}{4}\partial_\theta r \, \hat{\theta}\,,
    \label{eq:B2gradbeta}
\end{equation}
explicitly confirming (\ref{eq:betaboundary}).

\subsection{Hypertetrahedral model}
Finally, we turn to a line defect inside of a bulk model with hypertetrahedral symmetry $S_{N+1}\times Z_2$. As noted before, multiple stable fixed points have been found in these theories for $N=7,9$\cite{Pannell:2023pwz}, suggesting that the hypersurface $\partial \mathcal{M}$ permits a significantly richer fixed point structure than in the case of $O(N)$ or hypercubic bulks. Here, the bulk interaction tensor is most naturally formulated in terms of the $N+1$ vectors $\{e_i^\alpha\}$ forming the vertices of the hypertetrahedron, which satisfy the relations
\begin{equation}
    \sum_\alpha e_i^\alpha=0\,,\qquad \sum_\alpha e_i^\alpha e_j^\alpha=\delta_{ij}\,,\qquad e_i^\alpha e_i^\beta=\delta^{\alpha\beta}-\frac{1}{N+1}\,.
    \label{eq:evectorrules}
\end{equation}
These vectors can be explicitly constructed iteratively in $N$, using the rules
\begin{equation}
\begin{split}
    &(e_N)^\alpha_i=(e_{N-1})^\alpha_i\,\qquad i=1,\ldots,N\,, \,\alpha=1,\ldots, N\,, \\
    &(e_N)^\alpha_N=-\frac{1}{\sqrt{N(N+1)}}\,\qquad \alpha=1,\ldots,N\,, \\
    &(e_N)^{N+1}_i=\sqrt{\frac{N}{N+1}}\delta_{iN}\,.
\end{split}
\label{eq:edef}
\end{equation}
The $S_{N+1}$ symmetry acts on coupling space via the permutation of these vectors. For generic $N$, there are two bulk fixed points with this symmetry\cite{Osborn:2017ucf}
\begin{equation}
\begin{split}
    \lambda^+_{ijkl}=& \frac{\varepsilon}{3(N^2-5N+8)}\left(\delta_{ij}\delta_{kl}+\delta_{ik}\delta_{kl}+\delta_{il}\delta_{jk}\right)+\frac{(N-4)(N+1)\varepsilon}{3(N^2-5N+8)}\sum_\alpha e_i^\alpha e_j^\alpha e_k^\alpha e_l^\alpha\,,\\
    \lambda^-_{ijkl}=& \frac{\varepsilon}{3(N+3)}\left(\delta_{ij}\delta_{kl}+\delta_{ik}\delta_{kl}+\delta_{il}\delta_{jk}\right)+\frac{(N+1)\varepsilon}{3(N+3)}\sum_\alpha e_i^\alpha e_j^\alpha e_k^\alpha e_l^\alpha\,.
\end{split}
\end{equation}
For $N=4$ $\lambda^+_{ijkl}$ is equivalent to the hypercubic fixed point and for $N=5$ these two fixed points will coincide. For $N>5$ $\lambda^-_{ijkl}$ will be the stable of the two fixed points. The purely radial component of these tensors is
\begin{equation}
\begin{split}
    \lambda^+_{rrrr}=& \frac{\varepsilon}{N^2-5N+8}+\frac{(N-4)(N+1)\varepsilon}{3(N^2-5N+8)}\sum_\alpha \left(e_i^\alpha\hat{h}_i\right)^4 \,,\\
    \lambda^-_{rrrr}=& \frac{\varepsilon}{(N+3)}+\frac{(N+1)\varepsilon}{3(N+3)}\sum_\alpha \left(e_i^\alpha\hat{h}_i\right)^4\,,
\end{split}
\end{equation}
so that the boundary of $\mathcal{M}$ will be defined in each case by
\begin{equation}
    \begin{split}
        r^+(\vec{\Omega})=\left(\frac{1}{3(N^2-5N+8)}+\frac{(N-4)(N+1)}{9(N^2-5N+8)}\sum_\alpha \left(e_i^\alpha\hat{h}_i\right)^4\right)^{-1/2}\,,\\
        r^-(\vec{\Omega})=\left(\frac{1}{3(N+3)}+\frac{(N+1)}{9(N+3)}\sum_\alpha \left(e_i^\alpha\hat{h}_i\right)^4\right)^{-1/2}\,.
    \end{split}
    \label{eq:rhypertet}
\end{equation}
The function we must extremize to determine stability in both cases is $f(\vec{\Omega})=\sum_\alpha \left(e_i^\alpha\hat{h}^i\right)^4$, though there is once again a subtlety associated with the sign of the prefactor. One sees that the prefactor in $r^+$ will be negative for $N<4$, so that stable fixed points correspond to maxima of $f$. As in the hypercubic case, one can see that this maximum will occur precisely when we take $\hat{h}$ to point along one of the vectors $e^\alpha$. For $N>4$ the prefactor will be positive and stable fixed points will be constructed from vectors on the unit sphere minimizing $f$ for both hypertetrahedral bulks.

From the previous analysis of fixed points in the hypercubic model, it is natural to expect that the stable fixed points will lie in the center of the various hyperfaces of the hypertetrahedron. For that reason, let us consider as our ansatz on the unit sphere the vector
\begin{equation}
\textbf{h}^m_i=\sqrt{\frac{N+1}{m(N+1-m)}}\sum_{\alpha=1}^{m}e^{\alpha}_i\,,
\label{eq:ansatzm}
\end{equation}
where the first identity in (\ref{eq:evectorrules}) allows us to take this number of vectors in this sum to be $m\leq (N+1)/2$. It is important to note that this vector will not directly give a fixed point of the beta function. Instead, to get a fixed point we must rescale by $r(\vec{\Omega})|_{\textbf{h}^m}$ to move from the unit sphere to the surface $\mathcal{M}$. It is also important to note that the $S_{N+1}$ symmetry of the underlying bulk theory makes this choice is equivalent to choosing any $m$ of the vectors. To see that (\ref{eq:ansatzm}) is indeed a minimum of $f$, let us consider rotating this vector along the great circle spanned by $\textbf{h}^m$ and one of the $e^\alpha$ vectors. As $\{e^\alpha\}$ form a spanning set for the tangent space, these choices will be sufficient to explore the extremization of $f$. There are only two distinct choices one may make for $\alpha$: either $\alpha\leq m$ and $e^\alpha$ is in the sum, or $\alpha>m$. Let us first consider the former choice, taking $\alpha=1$. The great circle can be constructed explicitly using the Gram-Schmidt procedure on the set $\{\textbf{h}^m,\,e^1\}$ and the identities (\ref{eq:evectorrules}),
\begin{equation}
    \hat{h}_i(\theta)=\cos{\theta}\lsp\lsp\textbf{h}^m_i+\frac{\sin{\theta}}{\sqrt{1-\frac{N-m+1}{m \lsp N}}}\left(\sqrt{\frac{N+1}{N}}\lsp e^1_i-\sqrt{\frac{N-m+1}{m\lsp N}}\textbf{h}^m_i\right)\,.
\end{equation}
To construct $f$, we note that
\begin{equation}
    e^\alpha_i\hat{h}_i(\theta)=\begin{cases} \sqrt{\frac{N+1-m}{m(N+1)}}\cos{\theta}+\sqrt{\frac{m-1}{m}}\sin{\theta}\,, &  \alpha=1\\
     \sqrt{\frac{N+1-m}{m(N+1)}}\cos{\theta}-\frac{1}{\sqrt{m(m-1)}}\sin{\theta}\,, & 1<\alpha\leq m\\
     -\sqrt{\frac{m}{(N+1)(N+1-m)}}\cos{\theta}\,, & \alpha>m
\end{cases}\,,
\end{equation}
from which one finds that
\begin{equation}
    \frac{d f}{d \theta}\bigg|_{\theta=0}=0\,,\qquad\qquad\frac{d^2 f}{d \theta^2}\bigg|_{\theta=0}=\frac{8N-12m+8}{m(N+1-m)}\geq 0\,,
\end{equation}
where the last inequality follows from the restriction $m\leq (N+1)/2$. Thus, $\textbf{h}^m$ is stable with respect to perturbations that remain within the hyperface spanned by $\{e^1,\ldots,e^m\}$. To see what happens when moving outside of the hyperface, let us consider rotating towards $e^{N+1}$, so that $\hat{h}(\theta)$ now takes the form
\begin{equation}
    \hat{h}_i(\theta)=\cos{\theta}\lsp\lsp\textbf{h}^m_i+\frac{\sin{\theta}}{\sqrt{1-\frac{m}{N(N-m+1)}}}\left(\sqrt{\frac{N+1}{N}}\lsp e^{N+1}_i+\sqrt{\frac{m}{N(N-m+1)}}\textbf{h}^m_i\right)\,.
\end{equation}
The relevant dot products with the $e^\alpha$ vectors are now
\begin{equation}
    e^\alpha_i\hat{h}_i(\theta)=\begin{cases}
     -\sqrt{\frac{N+1-m}{m(N+1)}}\cos{\theta}\,, & \alpha\leq m \\
     \sqrt{\frac{m}{(N+1)(N+1-m)}}\cos{\theta}-\frac{1}{\sqrt{(N-m)(N+1-m)}}\sin{\theta}\,, & m<\alpha<N+1\\
     \sqrt{\frac{m}{(N+1)(N+1-m)}}\cos{\theta}+\sqrt{\frac{N-m}{N+1-m}}\sin{\theta}\,, &  \alpha=N+1
\end{cases}\,,
\end{equation}
from which we find
\begin{equation}
    \frac{d f}{d \theta}\bigg|_{\theta=0}=0\,,\qquad\qquad\frac{d^2 f}{d \theta^2}\bigg|_{\theta=0}=\frac{4(3m-N-1)}{m(N+1-m)}\,.
    \label{eq:hmstab}
\end{equation}
As $\partial_\theta f|_{\theta=0}=0$ in all directions, $\textbf{h}^m$ will be an extremal point of $f$ for all choices of $m$. For all $N$ we thus have a family of fixed points, $h^m_i=r(\textbf{h}^m)\textbf{h}^m_i$. The stability of $h^m_i$ is determined by the sign of the numerator in the second expression in (\ref{eq:hmstab}). For $3m>N+1$, $\partial^2_\theta f|_{\theta=0}>0$, so that $\textbf{h}^m$ will locally maximize $r(\vec{\Omega})$ in all directions and hence be totally stable. For $N+1>3m$, however, $\partial^2_\theta f|_{\theta=0}<0$, and instead $\textbf{h}^m$ will be unstable with respect to perturbations out of the hyperface\footnote{The case $N+1=3m$ is somewhat special. There exists an operator which is exactly marginal to all orders in perturbation theory corresponding to the existence of a centre manifold going through the fixed point. This marginal operator does not span a conformal manifold but instead indicates that the linear approximation breaks down in this direction. As the non-linear terms are cubic, this fixed point will be unstable.}.

At a given $N$, we have $\text{floor}\left(\frac{N+1}{2}\right)$ choices of $m$ that we might take, but the resulting fixed point will only be stable if $N+1>3m$. One can see that this will establish a lower bound on the number of stable fixed points that exist for a given $N$, given by
\begin{equation}
\#\text{ of stable fixed points}\geq\text{floor}\left(\frac{N+1}{2}\right)-\text{floor}\left(\frac{N+1}{3}\right)\,.
\label{eq:tetstablebound}
\end{equation}
The number of stable fixed points then grows at least as fast as $N$. For $N=7$ and $N=9$ this number is two, which matches the two stable fixed points found previously in \cite{Pannell:2023pwz}. To verify that the points found in that paper are in family $\{\textbf{h}^m\}$ let us consider the $N=7$ case, with the analysis for $N=9$ being identical. For $N=7$ there will be eight $e^\alpha$, so that we can take $m\in\{1,2,3,4\}$. Of these choices, only the latter two will be stable, so that we must consider the fixed points
\begin{equation}
    \textbf{h}^3_i=\sqrt{\frac{8}{15}}\left(e^6_i+e^7_i+e^8_i\right)\,,\qquad \textbf{h}^4_i=\frac{1}{\sqrt{2}}\left(e^5_i+e^6_i+e^7_i+e^8_i\right)\,.
\end{equation}
Matching the analysis in \cite{Pannell:2023pwz}, we can characterize these fixed points using the eigenvalues, $\kappa$, of the stability matrix and the invariant quantity
\begin{equation}
    H=-\frac{\varepsilon}{4}h_ih_i+\frac{1}{24}\lambda_{ijkl}h_i h_jh_kh_l\,.
\end{equation}
These invariants depend on the choice of $\lambda^-$ or $\lambda^+$ for the bulk interaction tensor, which we will consider separately. Choosing $\lambda^+$, we find that the value of $r^+(\vec{\Omega})$ at these points are given by $3\sqrt{55/17}$ and $\sqrt{33}$ respectively, so that the fixed points will be
\begin{equation}
    h^3_i=2\sqrt{\frac{66}{17}}\left(e^6_i+e^7_i+e^8_i\right)\,,\qquad h^4_i=\sqrt{\frac{33}{2}}\left(e^5_i+e^6_i+e^7_i+e^8_i\right)\,.
\end{equation}
The first of these has the invariants
\begin{equation}
    H=-\frac{495}{136}\,,\qquad \kappa\in\{1, \frac{14}{17}, \frac{14}{17}, \frac{2}{17}, \frac{2}{17}, \frac{2}{17}, \frac{2}{17}\}\,,
\end{equation}
while the second has the invariants
\begin{equation}
    H=-\frac{33}{8}\,,\qquad \kappa\in\{1,\frac{1}{2},\frac{1}{2},\frac{1}{2},\frac{1}{2},\frac{1}{2},\frac{1}{2}\}\,,
\end{equation}
which match (3.34) and (3.35) in \cite{Pannell:2023pwz} respectively. Choosing $\lambda^-$, the fixed points will be
\begin{equation}
    h^3_i=\frac{3\sqrt{5}}{2}\left(e^6_i+e^7_i+e^8_i\right)\,,\qquad h^4_i=\frac{3\sqrt{5}}{2}\left(e^5_i+e^6_i+e^7_i+e^8_i\right)\,.
\end{equation}
The first of these has the invariants
\begin{equation}
    H=-\frac{675}{256}\,,\qquad \kappa\in\{1, \frac{7}{16}, \frac{7}{16}, \frac{1}{16}, \frac{1}{16}, \frac{1}{16}, \frac{1}{16}\}\,,
\end{equation}
while the second has the invariants
\begin{equation}
    H=-\frac{45}{16}\,,\qquad \kappa\in\{1,\frac{1}{4},\frac{1}{4},\frac{1}{4},\frac{1}{4},\frac{1}{4},\frac{1}{4}\}\,,
\end{equation}
which again match (3.36) and (3.37) in \cite{Pannell:2023pwz} respectively, verifying that our method of computation produces the correct results. 

In fact, we believe that (\ref{eq:tetstablebound}) is likely saturated, that is to say that all stable fixed points in this theory lie within the family $\{\textbf{h}^m\}$. To motivate that these are the only possible stable points in this theory, let us consider analytic constructions for other fixed points and examine the manner in which they are unstable. As $f$ is quartic in the $e^\alpha$, fixed points will be vectors extremizing distance from both $\{e^\alpha\}$ and $\{-e^\alpha\}$. It is thus natural to consider an extension of (\ref{eq:ansatzm}) to include $-e^\alpha$ in the sum. We take the ansatz
\begin{equation}
    \mathfrak{h}^m_i=\frac{1}{\sqrt{2m}}\left(\sum_{\alpha=1}^m e^\alpha_i-\sum_{\alpha=m+1}^{2m} e^{\alpha}_i\right)\,,
\end{equation}
where $m<(N+1)/2$. The choice $m=(N+1)/2$ is disallowed here because one could use the first rule in (\ref{eq:evectorrules}) to write $\mathfrak{h}^{\frac{N+1}{2}}=\textbf{h}^{\frac{N+1}{2}}$, which was considered previously. We will again check whether or not these are extremum of $f$ by rotating towards the various $e^\alpha$. There are now three non-equivalent choices of $\alpha$, corresponding to whether $e^\alpha$ is in the first sum, the second, or in neither. As a representative of the first case, let us consider rotating towards $e^1$, so that the unit vector $\hat{h}$ will lie on the circle
\begin{equation}
    \hat{h}_i(\theta)=\cos{\theta}\lsp\lsp \mathfrak{h}^m_i+\frac{\sin{\theta}}{\sqrt{1-\frac{N+1}{2m\lsp N}}}\lsp\left(\sqrt{\frac{N+1}{N}}e^1_i-\sqrt{\frac{N+1}{2m\lsp N}}\mathfrak{h}^m_i\right)\,.
\end{equation}
The dot products of $\hat{h}(\theta)$ with the $e^\alpha$ vectors are then given by
\begin{equation}
    e^\alpha_i\hat{h}_i(\theta)=\begin{cases} \frac{1}{\sqrt{2m}}\cos{\theta}+\sqrt{\frac{2m\lsp N-N-1}{2m(N+1)}}\sin{\theta}\,, &  \alpha=1\\
     \frac{1}{\sqrt{2m}}\cos{\theta}+\frac{N+1-2m}{\sqrt{2m(1+N)(2m\lsp N-N-1)}}\sin{\theta}\,, & 1<\alpha\leq m\\
     -\frac{1}{\sqrt{2m}}\cos{\theta}-\frac{N+1-2m}{\sqrt{2m(1+N)(2m\lsp N-N-1)}}\sin{\theta}\,, & m<\alpha\leq 2m \\
     -\frac{1}{\sqrt{N(1+N)(1+\frac{N+1}{2m\lsp N})}}\sin{\theta}\,, & \alpha>2m
\end{cases}\,,
\end{equation}
and one finds that
\begin{equation}
    \frac{d f}{d \theta}\bigg|_{\theta=0}=0\,,\qquad\qquad\frac{d^2 f}{d \theta^2}\bigg|_{\theta=0}=\frac{4(6m^2-(N+1)^2+m(2N^2-N-3))}{m(N+1)((2m-1)N-1)}\geq0\,.
    \label{eq:hfrakstab1}
\end{equation}
One can check that (\ref{eq:hfrakstab1}) is also what one obtains when rotating with respect to a vector in the second sum, e.g. $e^{m+1}$, so that $\mathfrak{h}^m$ will be stable with respect to perturbations remaining within the hyperface spanned by $\{e^1,\ldots,e^m\}$ or $\{e^{m+1},\ldots,e^{2m}\}$. For rotations outside of either face, we can choose the vector $e^{N+1}$, for which one notes the property $e^{N+1}_i\mathfrak{h}^m_i=0$. The great circle then takes the simple form
\begin{equation}
    \hat{h}_i(\theta)=\cos{\theta}\lsp\lsp\mathfrak{h}^m_i+\sqrt{\frac{N+1}{N}}\sin{\theta}\lsp\lsp e^{N+1}_i\,,
\end{equation}
and the dot products are now
\begin{equation}
    e^\alpha_i\hat{h}_i(\theta)=\begin{cases} \frac{1}{\sqrt{2m}}\cos{\theta} -\frac{1}{\sqrt{N(N+1)}}\sin{\theta}\,, &  \alpha\leq m\\
     -\frac{1}{\sqrt{2m}}\cos{\theta} -\frac{1}{\sqrt{N(N+1)}}\sin{\theta}\,, & m<\alpha\leq 2m\\
     -\frac{1}{\sqrt{N(N+1)}}\sin{\theta}\,, & 2m<\alpha<N+1 \\
     \sqrt{\frac{N}{N+1}}\sin{\theta}\,, & \alpha=N+1
\end{cases}\,.
\end{equation}
A straightforward computation then yields
\begin{equation}
    \frac{d f}{d \theta}\bigg|_{\theta=0}=0\,,\qquad\qquad\frac{d^2 f}{d \theta^2}\bigg|_{\theta=0}=-\frac{2}{m}+\frac{12}{N(N+1)}\leq \frac{4(3-N)}{N(N+1)}\,,
    \label{eq:hfrakstab2}
\end{equation}
which confirms that $\mathfrak{h}^m$ will correspond to a fixed point of the beta function. As $N\geq3$ in cases of interest\footnote{The case of $N=2$ is in fact equivalent to the $O(2)$ model at one-loop, and $\partial_\theta^2f|{\theta=0}$ will vanish.}, this last expression will be non-positive, so that $\mathfrak{h}^m$ will be unstable with respect to perturbations moving out of the faces. 

There is an additional family of fixed points which can be constructed by noting that if we begin at the point $\textbf{h}^m$ and move off of the hyperedge in the direction of $e^{m+1}$ we will always encounter a second extremum of the same type as we move along the great circle. A third fixed point of the opposite stability must then lie between them. If $\textbf{h}^m$ is unstable with respect to pertubations out of the edge, it will correspond to a maximum of $f$ along the great circle spanned by $\{\textbf{h}^m,\,e^{m+1}\}$. As there is another maximum along this circle at the point $-e^{m+1}$, there must be a mimimum at $\textbf{h}^m-\tau\lsp e^{m+1}$ for some $\tau>0$. Similarly, if $\textbf{h}^m$ is stable, it will correspond to a minimum, so that there must be a maximum $\textbf{h}^m+\tau\lsp e^{m+1}$ for some $\tau>0$ before one reaches the second minimum at $\textbf{h}^{m+1}$. One can generalize this to include the fixed points $\textbf{h}^m$ more generally by considering the ansatz
\begin{equation}
    \mathfrak{h}^{m,m'}_i(\tau)=\sqrt{\frac{N+1}{m(N+1-m-m'\tau)+m'\tau((N+1-m')\tau-m)}}\left(\sum_{\alpha=1}^me^{\alpha}_i+\tau\sum_{\alpha=m+1}^{m+m'}e^{\alpha}_i\right)\,,
\end{equation}
where $m+m'<N+1$ and again we can use (\ref{eq:evectorrules}) to take $m,m'\leq (N+1)/2$. The point $\tau=1$ must be a fixed point, as it will correspond to $\textbf{h}^{m+m'}$, but we will show that there exists another non-zero value for $\tau$ which extremizes $f$. To construct $f(\tau)$ explicitly, we note that the dot products of $\hat{h}(\tau)$ with the various $e^\alpha$ vectors now take the form
\begin{equation}
    e^\alpha_i\hat{h}_i(\tau)=\begin{cases}
    \sqrt{\frac{N+1}{m(N+1-m-m'\tau)+m'\tau((N+1-m')\tau-m)}}\frac{N+1-m-m'\tau}{N+1}\,, & \alpha\leq m\\
     \sqrt{\frac{N+1}{m(N+1-m-m'\tau)+m'\tau((N+1-m')\tau-m)}}\frac{(N+1-m')\tau-m}{N+1}\,, & m<\alpha\leq m+m' \\
     -\sqrt{\frac{N+1}{m(N+1-m-m'\tau)+m'\tau((N+1-m')\tau-m)}}\frac{m+m'\tau}{N+1}\,, & \alpha>m+m'
\end{cases}\,.
\end{equation} 
One then finds that the derivative of $f$ with respect to $\tau$ is
\begin{equation}
    \frac{df}{d\tau}=-\frac{4mm'(m+m'-N-1)(1+N)(\tau-1)\tau\left((3m'-N-1)\tau+3m-N-1\right)}{(m^2+m'(m'-1-N)\tau-m(1+N-2m'\tau))^3}\,,
\end{equation}
which vanishes only for $\tau=0$, $1$ or the non-trivial solution
\begin{equation}
    \tau=-\left(\frac{3m-N-1}{3m'-N-1}\right)\,.
    \label{eq:tauchoice}
\end{equation}
It is important to note here that the sign of $\tau$ corresponds to whether or not $\textbf{h}^m$ and $\textbf{h}^{m'}$ have the same stability properties. If they are both stable or both unstable, then $\tau<0$, while if they are of opposite stability, then $\tau>0$. This is illustrated in figure \ref{fig:stabfig}, where we have plotted $f(\tau)$ for $N=9$ and three different choices of $m$ and $m'$.
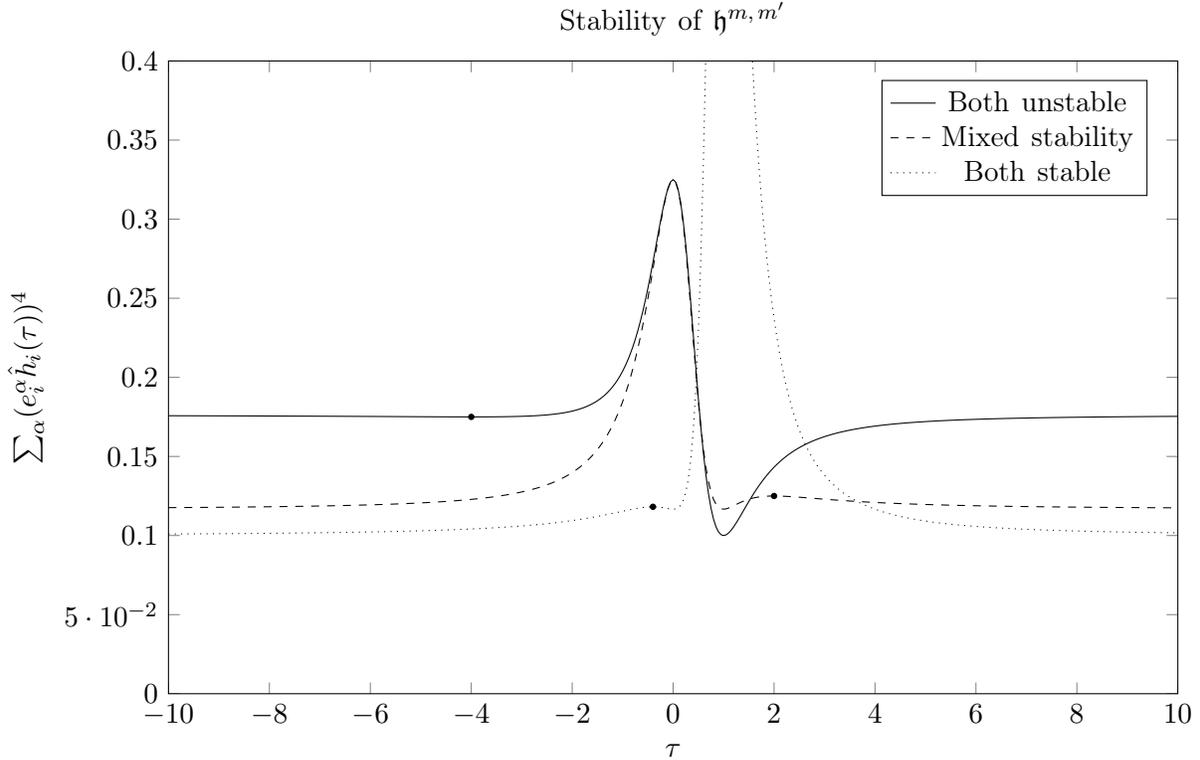
\begin{figure}[t!]
\centering
\begin{tikzpicture}
\begin{axis}[width=15cm,height=10cm,
    xmin=-10, xmax=10,
    ymin=0, ymax=0.4,
    xlabel= $\tau$,
    ylabel= $\sum_\alpha(e^\alpha_i\hat{h}_i(\tau))^4$,
    title={Stability of $\mathfrak{h}^{m,\,m'}$},
    legend pos = north east,
]
\addplot+[
    no marks,color=black]
table{datafiles/fuu.dat};
\addplot+[
    no marks,dashed,color=black]
table{datafiles/fus.dat};
\addplot+[
    no marks,dotted,color=black]
table{datafiles/fss.dat};
\addplot+[only marks, mark=*,mark options={color=black,fill=black},mark size=1pt] coordinates{(-4, 7/40) (2,1/8) (-2/5,13/110)};
\legend{Both unstable, Mixed stability, Both stable}
\end{axis}
\end{tikzpicture}
    \caption{A plot of $f(\tau)$ with $N=9$ for various choices of $m$ and $m'$. The solid line indicates the choice $m=2$, $m'=3$ where both $\textbf{h}^m$ and $\textbf{h}^{m'}$ are unstable, the dashed line indicates the choice $m=2$, $m'=4$ where $\textbf{h}^m$ is unstable and $\textbf{h}^{m'}$ is stable, and the dotted line indicates the choice $m=4$, $m'=5$ where both $\textbf{h}^m$ and $\textbf{h}^{m'}$ are stable. The black dots indicate the non-trivial fixed points $\mathfrak{h}^{m,m'}$ with $\tau$ given by (\ref{eq:tauchoice}).}
    \label{fig:stabfig}
\end{figure}
Features one can observe in this figure with regards to the stability of $\mathfrak{h}^{m,m'}$ will hold more generally. If both $\textbf{h}^m$ and $\textbf{h}^{m'}$ are unstable, $\tau=0$ and $\tau\rightarrow\pm\infty$ must correspond to maxima of $f(\tau)$, so that there must exist at least two minima, one with $\tau>0$ and one with $\tau<0$. As noted above, the point $\tau=1$ will give $\textbf{h}^{m+m'}$, which from the above analysis one will see is stable with respect to deformations towards either $\textbf{h}^m$ or $\textbf{h}^{m'}$ and is thus our minimum with $\tau>0$. As $\mathfrak{h}^{m,m'}$ is the only other fixed point, it must minimize $f(\tau)$. Similarly one can see, as illustrated in the figure, that if $\textbf{h}^m$ and $\textbf{h}^{m'}$ are stable or if they are of opposite stability, $\mathfrak{h}^{m,m'}$ must lie between two minima, and thus maximize $f(\tau)$. We see that $\mathfrak{h}^{m,m'}$ is only stable with respect to rotations along the $\textbf{h}^m$-$\textbf{h}^{m'}$ great circle if both $\textbf{h}^m$ and $\textbf{h}^{m'}$ correspond to unstable fixed points.

One may verify by rotating towards each of the $e^\alpha$ that (\ref{eq:tauchoice}) will indeed correspond to a fixed point of the beta function for any choice of $m$ and $m'$. To demonstrate that $\mathfrak{h}^{m,m'}$ is not a minimum of $f$ for any $m,m'\geq1$, let us consider rotating this solution towards the vector $e^{N+1}$, which we can assume is in neither sum. Constructing the great circle spanned by $\{
\mathfrak{h}^{m,m'},\,e^{N+1}\}$ as before, one finds that
\begin{equation}
    \frac{d^2f}{d\theta^2}\bigg|_{\theta=0}\propto-(3m-N-1)(3m'-N-1)
\end{equation}
where in the interest of space we have not included a multiplicative factor which is positive for all values of $N$, $m$ and $m'$. One sees that $\mathfrak{h}^{m,m'}$ will be locally maximize $f$ along this great circle when $m,m'<(N+1)/3$. Combining this with the previous analysis, we find that that no choices of $m$ and $m'$ will produce minima of $f$. Thus, all fixed points of the beta function corresponding to the family $\mathfrak{h}^{m,m'}$ will be unstable\footnote{In fact if $m$ and $m'$ are such that $3m-N-1=N+1-3m'$ then the resulting point will be stable. However as $\tau=1$ here, this is really just the fixed point $\textbf{h}^{m+m'}$ which was considered previously. As this is not a new fixed point, we will not include it in the family $\mathfrak{h}^{m,m'}$.}.

For $N\leq9$, numerical results indicate that all fixed points are classified by the families $\textbf{h}^m$, $\mathfrak{h}^m$ and $\mathfrak{h}^{m,m'}$. One is thus able to definitively say that all stable fixed points live within the family $\textbf{h}^m$ for these $N$, as all other fixed points are explicitly unstable. For larger values of $N$ it is possible that other, more complicated families of fixed points arise, for instance by generalizing $\mathfrak{h}^{m,m'}$ to include three or more terms. However, the manner in which the families $\mathfrak{h}^m$ and $\mathfrak{h}^{m,m'}$ are unstable is instructive. One can convince oneself that any point on the unit sphere can be written in the form
\begin{equation}
    \hat{h}_{i}=\sum_{\alpha=1}^{N+1}h_{\alpha}e^{\alpha}_i\,,
\end{equation}
where using (\ref{eq:evectorrules}) we can take $h_{\alpha}\geq0$ $\forall\alpha$. Representing $\mathfrak{h}^m$ and $\mathfrak{h}^{m,m'}$ in this way, one finds that the unstable direction corresponds to rotating towards $e^{\alpha'}$ with $h^{\alpha'}$ being the smallest non-zero coefficient. Geometrically, this rotation moves the vector towards either the center of a hyperface or towards the center of the hyperedge $h^{\alpha'}=0$. We expect that this intuition holds more generally, so that if there exists a fixed point with not all the $h_\alpha$ the same, it will be possible to decrease $f$ by moving towards the center of the nearest hyperface. Thus, the only stable fixed points at any value of $N$ should come from the family $\textbf{h}^m$, which lies at these centers.

\begin{figure}[t!]
    \centering
\begin{tikzpicture}
\begin{axis}[
    xmin=-5, xmax=5,
    ymin=-5, ymax=5,
    xlabel= $h_3$,
    ylabel= $h_4$,
    ylabel style={rotate=-90},
    xticklabel style={scaled ticks=false, /pgf/number format/fixed, /pgf/number format/precision=3},
    yticklabel style={scaled ticks=false, /pgf/number format/fixed, /pgf/number format/precision=3},
    title={Solution validity for a line defect in a $S_5\times\mathbb{Z}_2$ bulk},
    legend pos = outer north east,
    view={0}{90},
    domain=-5:5
]
\addplot+[
    no marks,color=black,fill=red!90!black,opacity=0.5]
table{datafiles/T4.dat};
\addplot+[only marks,mark=*,mark options={color=black,fill=black}, ultra thick] coordinates{(0,3.1749) (0,-3.1749) (3.07409, -0.793725) (-3.07409, 0.793725)};
\addplot+[only marks, mark=square*,mark options={color=black,fill=black}, ultra thick] coordinates{(-2.07255, 2.67565) (2.07255, -2.67565)};
\addplot+[only marks, mark=diamond*,mark options={color=black,fill=black}, ultra thick] coordinates{(3.07409, 2.38118) (-3.07409, -2.38118)};
\addplot+[only marks, mark=x,mark options={color=black,fill=black}, thick] coordinates{(0,0)};
\addplot3 [
        gray,-stealth,samples=20,
        quiver={
            u={-((x* (-756 + 71*x^2 - 6*sqrt(15)*x*y + 45*y^2))/1512)/sqrt(((x* (-756 + 71*x^2 - 6*sqrt(15)*x*y + 45*y^2))/1512)^2+((-2*sqrt(15)*x^3 - 756*y + 45*x^2*y + 75*y^3)/1512)^2))},
            v={-((-2*sqrt(15)*x^3 - 756*y + 45*x^2*y + 75*y^3)/1512)/sqrt(((x* (-756 + 71*x^2 - 6*sqrt(15)*x*y + 45*y^2))/1512)^2+((-2*sqrt(15)*x^3 - 756*y + 45*x^2*y + 75*y^3)/1512)^2))},
            scale arrows=0.25,
        },
    ] (x,y,0);
\end{axis}
\end{tikzpicture}
    \caption{Perturbative defect RG flows within at $T_4^-$ bulk. For simplicity we display only the two-dimensional invariant subspace $\{h_i|h_1=h_2=0\}$. The red region indicates the region $\mathcal{M}$. Beyond this region $A$ will monotonically increase. There are three non-trivial fixed points modulo the action of $S_5\times\mathbb{Z}_2$. There appear to be two types of stable fixed points, indicated by diamonds and squares, but considering the situation in the full four-dimensional coupling space one finds that the squares lie on saddle points, so that only the diamonds are stable against all deformations.}
    \label{fig:T4metric}
\end{figure}
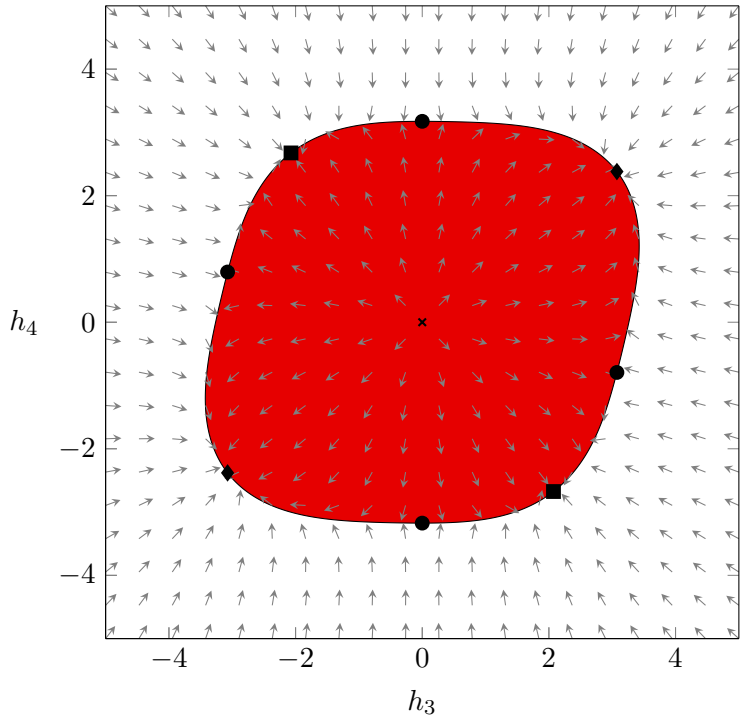

The situation for $N=4$, with the choice $\lambda^-$, is shown in Figure \ref{fig:T4metric}, where for simplicity we show only the two-dimensional invariant subspace $\{h_i|h_1=h_2=0\}$. As before, it is important to recognize that we must properly define flows by quotienting out the action of the symmetry group $S_{5}\times \mathbb{Z}_2$, so that the figure displays redundency. Taking polar coordinates, one finds that (\ref{eq:rhypertet}) implies that the boundary curve is given by
\begin{equation}
     r=6\sqrt{\frac{21}{66-2\cos{2\theta}+7\cos{4\theta}-2\sqrt{15}\sin{2\theta}-\sqrt{15}\sin{4\theta}}}\,.
 \end{equation}
 As expected one again finds that along the boundary
 \begin{equation}
    \vec{\beta}=-\frac{\varepsilon}{4}\partial_\theta r \, \hat{\theta}\,,
\end{equation}
so that flows beginning within $\mathcal{M}$ are confined to stay within the region. The circles indicate points in the class $\textbf{h}^1$, while the diamonds represent the class $\textbf{h}^2$. The squares instead lie in the class $\mathfrak{h}^1$. It is important to note that the stability of $\mathfrak{h}^1$ apparent in this figure is only with respect to deformations that remain within this two-dimensional subspace of the full four-dimensional coupling space. Considering the situation in the three-dimensional hyperplane $h_1=0$, one sees that the squares in fact lie at saddle points of the boundary manifold, so that there exist deformations which trigger a flow away from these points, as expected from the prior discussion about the family of fixed points $\mathfrak{h}^m$. This can be verified explicitly by examining the eigenvalues $\kappa$ of the stability matrix, which show two negative $\kappa$'s with eigenvectors living within the $h_1-h_2$ plane.

\section{Stability in scalar surface defects}\label{surfacesec}
Now, let us consider deforming the bulk model by the introduction of a $p=2$ dimensional surface defect. This will introduce a term into the action quadratic in the scalar field $\phi_i$
\begin{equation}
    S_{surface}=\int d^{2}\vec{x}\, h_{ij}\phi_i(\vec{x},\vec{0})\phi_j(\vec{x},\vec{0})\,,
\end{equation}
where $h_{ij}$ is symmetric. This deformation has been studied previously for an $O(N)$ bulk in \cite{Trepanier:2023tvb,Giombi:2023dqs,Krishnan:2023cff}, but one may easily extend those results to a generic bulk interaction tensor $\lambda_{ijkl}$ to find the beta function\footnote{Note that here we have rescaled $\lambda_{ijkl}$ and $h_{ij}$ to remove factors of $\pi$ from the beta function.}
\begin{equation}
    \beta_{ij}=-\varepsilon\lsp h_{ij}+h_{ik}h_{kj}+\lambda_{ijab}h_{ab}\,.
    \label{eq:betasurface}
\end{equation}
This beta function is gradient, and may be written as before in the form
\begin{equation}
    \beta_{ij}=-\frac{\varepsilon}{2}T_{ij,kl}\partial_{kl}r^2
\end{equation}
for the metric
\begin{equation}
    T_{ij,kl}=\frac{1}{2}(\delta_{ik}\delta_{jl}+\delta_{il}\delta_{jk})-\frac{1}{\varepsilon}\lambda_{ijkl}-\frac{1}{\varepsilon}h_{ik}h_{jl}\,,
\end{equation}
where we define the radius in coupling space to be $r=h_{ij}h_{ij}$. The dependence of this $A$ function upon the energy scale $\mu$ can be determined by the differential equation
\begin{equation}
    \frac{d A}{d\ln{\mu}}=\varepsilon r^2(\varepsilon-\lambda_{ijkl}\hat{h}_{ij}\hat{h}_{kl}-\hat{h}_{ij}\hat{h}_{ik}\hat{h}_{jk}r)\,,
\end{equation}
for vectors on the unit sphere $\hat{h}_{ij}=h_{ij}/r$. At fixed points $\frac{d A}{d\ln{\mu}}=0$ by definition, so that all non-trivial fixed points must lie on the surface defined by
\begin{equation}
    r=\frac{\varepsilon-\lambda_{ijkl}\hat{h}_{ij}\hat{h}_{kl}}{\hat{h}_{ij}\hat{h}_{ik}\hat{h}_{jk}}\,.
    \label{eq:rsurface}
\end{equation}
If the numerator of this expression is negative at a given fixed point $h^*_{ij}$, then
\begin{equation}
    h^*_{ij}S_{ij,kl}h^*_{kl}=\varepsilon h^*_{ij}h^*_{ij}-\lambda_{ijkl}h^*_{ij}h^*_{kl}<0\,,
\end{equation}
so that there must exist at least one negative eigenvalue and corresponding relevant operator at that point. In examining stable fixed points, we can thus restrict to the case $\varepsilon h^*_{ij}h^*_{ij}-\lambda_{ijkl}h^*_{ij}h^*_{kl}>0$. In fact, for such points a straightforward extension of Michel's theorem continues to hold.

\begin{theorem}
    If $h^*_{ij}$ is a stable fixed point of (\ref{eq:betasurface}), then it must be a global minimum of the function $$A=-\frac{\varepsilon}{2}h_{ij}h_{ij}+\frac{1}{2}\lambda_{ijkl}h_{ij}h_{kl}+\frac{1}{3}h_{ij}h_{jk}h_{ki}$$ in the space of fixed points\footnote{Note that global minimization in the space of fixed points can be recast as global minimization of $A$ restricted to the surface defined by (\ref{eq:rsurface}).}. If $$(u,v)_\lambda=u_{ij}v_{ij}-\lambda_{ijkl}u_{ij}v_{ij}$$ is a positive definite quadratic form then the stable defect must preserve the bulk symmetry.
\end{theorem}
\begin{proof}
    The proof is identical to the one presented in section \ref{michelsec}. The only complication is that the quadratic term in $A$ no longer takes the simple form $(h,h)=h_{ij}h_{ij}$, but is instead the inner product $(h,h)_{\lambda}=\varepsilon h_{ij}h_{ij}-\lambda_{ijkl}h_{ij}h_{kl}$. While for generic $\lambda_{ijkl}$ this may be negative at some fixed points, previous analysis tells us that these fixed points will already be unstable, so that we can restrict our attention to fixed points such that $(h^*,h^*)_{\lambda}>0$. As $A=-\frac{1}{6}(h^*,h^*)_\lambda$ at fixed points, points with negative $(h^*,h^*)_{\lambda}$ will have $A>0$, and thus not affect our conclusion of stable fixed points globally minimizing $A$.

    However, the difference in quadratic forms does impact the conclusions about fixed point uniqueness. In the proof of Michel's theorem we used that if $(v,v)=0$ then $v=0$ to show that if two fixed points have the same value of $A$ then they must both be unstable. As we can no longer a priori say that $(,)_\lambda$ is a positive definite form, it is possible for $(v,v)_\lambda=0$ to generically have non-trivial solutions. If we assume that $(v,v)_\lambda>0$ for all non-zero $v$, then proof follows as before and if $h^*_{ij}$ is stable then it must be unique. If a fixed point breaks the bulk symmetry group, then the broken group elements will move $h^*_{ij}$ within an orbit of distinct but physically equivalent fixed points. As $A$ is invariant under any group action, these points will all have the same value of $A$, so that if $(,)_\lambda$ is positive definite, we see that the only totally possible stable fixed point is one which preserves the entire bulk symmetry group.
\end{proof}

As with Michel's original theorem, this is applicable also to symmetry preserving deformations. By restricting to a symmetry preserving submanifold, what we mean by a unique fixed point may change. Fixed points breaking the bulk symmetry group to some subgroup $H\subset G$ will lie a conformal manifold corresponding to the orbit of the fixed point under $G$. A hyperplane of deformations preserving $H$ will only intersect this orbit at a single point, so that this fixed point appears to be unique from the perspective of the theorem applied to this hyperplane. 

As $\lambda_{ijkl}$ is a rank-4 totally symmetric tensor, it is always possible to decompose it into symmetry $O(N)$ representations like
\begin{equation}
    \lambda_{ijkl}=d_0\left(\delta_{ij}\delta_{kl}+\text{Perms.}\right)+\left(d_{2,ij}\delta_{kl}+\text{Perms.}\right)+d_{4,ijkl}\,,
\end{equation}
where $d_{2,ij}$ and $d_{4,ijkl}$ are symmetric and traceless. If the symmetry group $H\subseteq O(N)$ preserved by the tensor $\lambda_{ijkl}$ does not have a rank-2 invariant tensor other than $\delta_{ij}$, this middle term will vanish, as is often the case for bulk critical models\cite{Osborn:2017ucf,Osborn:2020cnf}. One then sees that the only $h_{ij}$ one can construct to preserve $H$ will be $h_{ij}=h \delta_{ij}$, for which the beta function is
\begin{equation}
    \beta_{ij}=\left(-\varepsilon +(N+2)d_0 +h\right)h_{ij}\,.
\end{equation}
The symmetry preserving defect in these cases is then explicitly given by
\begin{equation}
    h_{ij}=\left(\varepsilon-(N+2)d_0\right)\delta_{ij}\,,
    \label{eq:stabsurfacedef}
\end{equation}
where the value of $d_0$ depends upon the chosen bulk critical model. The above theroem then implies that if $\lambda_{ijkl}$ is such that
\begin{equation}
    d_{2,ij}=0\,,\quad\text{and}\quad \varepsilon\lsp h_{ij}h_{ij}-\lambda_{ijkl}h_{ij}h_{kl}>0\quad\forall h_{ij}\neq0\,,
\end{equation}
then the only possible totally stable defect is the one given in (\ref{eq:stabsurfacedef}). For the sake of completeness, let us now examine the application of this theorem for a number of bulk critical models.

\subsection{Free bulk}
Unlike the line defect, the presence of the defect-defect term in the beta function permits the existence of non-trivial surface dCFTs living inside of a free bulk. As $\lambda_{ijkl}=0$ here, the conditions $d_{2,ij}$ and $\varepsilon\lsp h_{ij}h_{ij}-\lambda_{ijkl}h_{ij}h_{kl}>0$ trivially hold, so that the theorem tells us that the only possible totally stable fixed point will be
\begin{equation}
    h_{ij}=\varepsilon \delta_{ij}\,.
\end{equation}
The stability matrix for this fixed point takes the simple form
\begin{equation}
    S_{ij,kl}=\frac{\varepsilon^2}{2}(\delta_{ik}\delta_{jl}+\delta_{il}\delta_{jk})\,,
\end{equation}
which is manifestly positive-definite for all values of $N$. The free bulk thus has the unique totally stable fixed point $h_{ij}=\varepsilon\delta_{ij}$ for all $N$. As the symmetry preserving defect is totally stable for all values of $N$, the theorem also guarantees that it will be the only fixed point stable under symmetry preserving deformations for any subgroup $H\subseteq O(N)$.

\subsection{\texorpdfstring{$O(N)$}{O(N)} model}
The simplest interacting critical bulk model is the generalization of the Wilson-Fischer fixed point to $N$ fields
\begin{equation}
    \lambda_{ijkl}=\frac{\varepsilon}{N+8}\left(\delta_{ij}\delta_{kl}+\text{Perms.}\right)\,,
\end{equation}
which remains invariant under generic $O(N)$ rotations. One sees immediately that $d_{2,ij}=0$ by definition, so that there will exist a symmetry-preserving fixed point
\begin{equation}
    h_{ij}=\frac{6\varepsilon}{N+8}\delta_{ij}\,.
    \label{eq:surfONstab}
\end{equation}
To see that this is the only possible stable fixed point, we must examine the positive-definiteness of the inner product $(h,h)_{\lambda}$. Using the explicit form of $\lambda_{ijkl}$ we have that
\begin{equation}
\begin{split}
    (h,h)_{\lambda}&=\frac{(N+6)\varepsilon}{N+8}h_{ij}h_{ij}-\frac{\varepsilon}{N+8}h_{ii}h_{jj} \\
    &=\frac{(N+6)\varepsilon}{N+8}\sum_{i\neq j}h_{ij}^2+\frac{(N+5)\varepsilon}{N+8}\sum_{i}h_{ii}^2-\frac{\varepsilon}{N+8}\sum_{i\neq j}h_{ii}h_{jj}\,.
\end{split}
\end{equation}
Completing the square using the last two terms, this is equal to
\begin{equation}
    (h,h)_\lambda=\frac{(N+6)\varepsilon}{N+8}\sum_{i\neq j}h_{ij}^2+\frac{6\varepsilon}{N+8}\sum_{i}h_{ii}^2+\frac{\varepsilon}{N+8}\sum_{i>j}(h_{ii}-h_{jj})^2>0\quad\forall N\,.
\end{equation}
One sees that (\ref{eq:surfONstab}) is the only possible stable defect in an $O(N)$ bulk. To see for which $N$ this point is indeed stable, we need to consider the stability matrix explicitly
\begin{equation}
    S_{ij,kl}=\frac{\varepsilon}{N+8}\delta_{ij}\delta_{kl}+\frac{(6-N)\varepsilon}{2(N+8)}(\delta_{ik}\delta_{jl}+\delta_{il}\delta_{jk})\,.
\end{equation}
Contracting with an arbitrary vector $v_{ij}$, this is
\begin{equation}
    v_{ij}S_{ij,kl}v_{kl}=\frac{(6-N)\varepsilon}{N+8}v_{ij}v_{ij}+\frac{\varepsilon}{N+8}v_{ii}v_{jj}\,.
\end{equation}
For $N\leq 6$ this is the positive sum of squares, and thus manifestly positive definite, but for $N>6$ there will become negative off-diagonal squares for $i\neq j$. We thus conclude that for $N\leq 6$ (\ref{eq:surfONstab}) will be the unique totally stable fixed point in an $O(N)$ bulk, and that for $N>6$ there will be no totally stable defects. We note that this analysis matches the conclusions previously reached in \cite{Trepanier:2023tvb,Giombi:2023dqs,Krishnan:2023cff}.

\subsection{Hypercubic model}
Another well-studied bulk critical model is the hypercubic model, where the interaction tensor
\begin{equation}
    \lambda_{ijkl}=\frac{\varepsilon}{3 N}\left(\delta_{ij}\delta_{kl}+\text{Perms.}\right)+\frac{N-4}{3N}\delta_{ijkl}
    \label{eq:lambdacubic}
\end{equation}
breaks $O(N)$ symmetry to the hypercubic group $B_N=\mathbb{Z}_2^N\rtimes S_N$. It is important to note that while $\delta_{ijkl}$ is symmetric, it is not traceless, so that the correct value of $d_0$ is not $1/3N$ but instead
\begin{equation}
    d_0=\frac{2(N-1)}{3N (N+2)}\varepsilon\,.
\end{equation}
Noting that $d_{2,ij}=0$, we have the symmetry preserving defect
\begin{equation}
    h_{ij}=\frac{(N+2)\varepsilon}{3N}\delta_{ij}\,.
    \label{eq:surfBNstab}
\end{equation}
To see that this is the only possible totally stable fixed point, we note that
\begin{equation}
\begin{split}
    (h,h)_{\lambda}&=\frac{(3N-2)\varepsilon}{3N}h_{ij}h_{ij}-\frac{\varepsilon}{3N}h_{ii}h_{jj}-\frac{(N-4)\varepsilon}{3N}\sum_i h_{ii}^2\\
    &=\frac{(3N-2)\varepsilon}{3N}\sum_{i\neq j}h_{ij}^2+\frac{(N+2)\varepsilon}{3N}\sum_{i}h_{ii}^2+\frac{\varepsilon}{3N}\sum_{i>j}(h_{ii}-h_{jj})^2\,.
\end{split}
\end{equation}
As $3N-2$ is positive when $N$ is a natural number, $(h,h)_{\lambda}>0$ for all non-zero $h$ so that the uniqueness theorem is applicable. As before, the theorem only guarantees that no other fixed points will be totally stable, so to determine the stability of the symmetry preserving defect we must examine the form of the stability matrix explicitly
\begin{equation}
    S_{ij,kl}=\left(\frac{\varepsilon}{N}-\frac{\varepsilon}{6}\right)(\delta_{ik}\delta_{jl}+\delta_{il}\delta_{jk})+\frac{\varepsilon}{3N}\delta_{ij}\delta_{kl}+\frac{N-4}{3N}\delta_{ijkl}\,.
\end{equation}
Contracting with two factors of an arbitrary vector $v_{ij}$ yields
\begin{equation}
    v_{ij}S_{ij,kl}v_{kl}=\left(\frac{2\varepsilon}{N}-\frac{\varepsilon}{3}\right)\sum_{i\neq j}v_{ij}^2+\frac{2\varepsilon}{3N}\sum_ih_{ii}^2+\frac{\varepsilon}{3N}(v_{ii})^2\,.
\end{equation}
For $N>6$ the prefactor of the first term will be negative, so that there will exist vectors such that $v_{ij}S_{ij,kl}v_{kl}<0$, which guarantees the existence of a relevant operator at the fixed point. We thus conclude that for $N\leq 6$ the symmetry preserving defect given by (\ref{eq:surfBNstab}) is the unique totally stable fixed point in a hypercubic bulk, and that for $N>6$ there are no totally stable fixed points at all.

\subsection{MN model}
It is possible to generalize the form of the hypercubic model to construct a bulk which is invariant under the symmetry group $O(m)^n\rtimes S_n$, where $N=mn$. Breaking up the $O(N)$ index into double indices $i_r$, where $i=1,\,\ldots,\, n$ and $r=1,\,\ldots,\, m$, the interaction tensor at the bulk critical point takes the form
\begin{equation}
    \lambda_{i_a j_{b} k_{c}l_{d}}=\frac{(4-m)\varepsilon}{m(N-16)+8(2+N)}\left(\delta_{i_a j_{b}}\delta_{k_{c}l_{d}}+\text{Perms.}\right)+\frac{(N-4)\varepsilon}{m(N-16)+8(2+N)}\delta_{abcd}\left(\delta_{ij}\delta_{kl}+\text{Perms.}\right)\,.
\end{equation}
Setting $m=1$ yields the interaction tensor for the hypercubic model given in (\ref{eq:lambdacubic}). Once again, $d_{2,ij}=0$, but the last term is not traceless, so that the value of $d_0$ must be calcluated using $d_0=\lambda_{iijj}/(N(N+2))$ to find
\begin{equation}
    d_0=\frac{6(N-m)}{(N+2)(m(N-16)+8(2+N))}\varepsilon\,.
\end{equation}
As one can check, setting $m=1$ in this expression yields the proper value for $d_0$ in the hypercubic model. For all values of $m$ and $N$ we will thus have a symmetry-preserving defect fixed point given by
\begin{equation}
    h_{ij}=\frac{(m(N-10)+2(8+N))\varepsilon}{m(N-16)+8(2+N)}\delta_{ij}\,.
    \label{eq:surfMNstab}
\end{equation}
In order to apply the theorem of uniqueness we once again study the positive-definiteness of $(h,h)_\lambda$ in this model
\begin{equation}
\begin{split}
    (h,h)_\lambda=&\frac{m(N-14)+8(N+1)}{m(N-16)+8(N+2)}\varepsilon h_{i_a j_b}h_{i_a j_b}-\frac{4-m}{m(N-16)+8(N+2)}\varepsilon(h_{i_a i_a})^2\\ &-\frac{N-4}{m(N-16)+8(N+2)}\varepsilon h_{i_a i_a}h_{j_a j_a}-\frac{2(N-4)}{m(N-16)+8(N+2)}\varepsilon h_{i_a j_a}h_{i_a j_a} \\
    =&\sum_{a\neq b, i,j}\frac{m(N-14)+8(N+1)}{m(N-16)+8(N+2)}\varepsilon h_{i_a j_b}^2+\sum_{a,i\neq j} \frac{m(N-14)+6N+16}{m(N-16)+8(N+2)}\varepsilon h_{i_a j_a}^2 \\ &+\sum_{a,i} \frac{16+m(N-13)+5 N}{m(N-16)+8(N+2)}\varepsilon h_{i_a i_a}^2-\sum_{a\neq b i,j}\frac{4-m}{m(N-16)+8(N+2)}\varepsilon h_{i_a i_a}h_{j_b j_b}\\ &-\sum_{a, i\neq j}\frac{N-m}{m(N-16)+8(N+2)}\varepsilon h_{i_a i_a} h_{j_a j_a}\,.
\end{split}
\end{equation}
In order to ensure that we obtain a positive sum of squares, the manner in which we complete the square depends upon the sign of $4-m$, but in both cases one finds that $(h,h)_{\lambda}>0$ for all non-zero $h$. Applying the theorem, we conclude that (\ref{eq:surfMNstab}) is the only possible totally stable fixed point. The uniqueness of the stable fixed point allows us to easily determine ranges of $m$ and $N$ for which there exist no totally stable fixed points at all. The stability matrix at the symmetry preserving fixed point is given by
\begin{equation}
\begin{split}
    S_{i_a j_b ,k_cl_d}=&\frac{(m-4)(N-6)}{2m(N-16)+8(N+2)}\varepsilon(\delta_{i_a k_c}\delta_{j_b l_d}+\delta_{il}\delta_{jk})+\frac{4-m}{m(N-16)+8(N+2)}\varepsilon\delta_{i_a j_b}\delta_{k_c l_d}\\&+\frac{N-4}{m(N-16)+8(N+2)}\delta_{abcd}(\delta_{ij}\delta_{kl}+\text{Perms.})\,.
\end{split}
\end{equation}
Contracting this with two factors of an arbitrary vector $v_{ij}$, one finds after some tedious but straightforward algebra that the stability of (\ref{eq:surfMNstab}) requires a number of conditions on $m$ and $N$. First, the positivity of off-diagonal $v_{ij}^2$ terms requires that
\begin{equation}
    m (N-6)-2 (N-8)\geq 0\quad \text{and} \quad (m-4) (N-6)\geq 0\,.
\end{equation}
One then finds after completing the square, that requiring all of the coefficients of the resulting terms be positive yields different conditions depending upon the size of $m$ and $N$. If $m>4$, then we must have
\begin{equation}
    (m-4) (2 m-N-4)\geq 0\,,
\end{equation}
while if $m\leq4$ we must have
\begin{equation}
    2 (m (N-5)-N+8)\geq0\quad \text{if }N\leq4\,,
\end{equation}
where we also assume by construction that $m<N$. The region satisfying these inequalities\footnote{Note that for $m=N$ $\lambda_{ijkl}=0$, so that the symmetry preserving defect will always be stable. For this value of $m$ there is a special cancellation between the off-diagonal terms which explains why points which appear to violate the bounds in fact remain stable.} are shown in Figure \ref{fig:MNregion}.
\begin{figure}[t!]
    \centering
\begin{tikzpicture}
\begin{axis}[
    xmin=1, xmax=20,
    ymin=1, ymax=10,
    xlabel= $N$,
    ylabel= $m$,
    ylabel style={rotate=-90},
    xticklabel style={scaled ticks=false, /pgf/number format/fixed, /pgf/number format/precision=3},
    yticklabel style={scaled ticks=false, /pgf/number format/fixed, /pgf/number format/precision=3},
    title={Region of stability for $MN$ symmetry preserving defect},
    legend pos = outer north east,
    view={0}{90}
]
\addplot+[
    no marks,color=black,fill=red!90!black,opacity=0.5]
table{datafiles/MNRegionUpper.dat};
\addplot+[
    no marks,color=black,fill=red!90!black,opacity=0.5]
table{datafiles/MNRegionLower.dat};
\end{axis}
\end{tikzpicture}
    \caption{The region of stability for the symmetry preserving defect in the MN model. Within the colored region the symmetry preserving defect is the unique stable fixed point, and outside of the colored region there exist no stable fixed points whatsoever.}
    \label{fig:MNregion}
\end{figure}
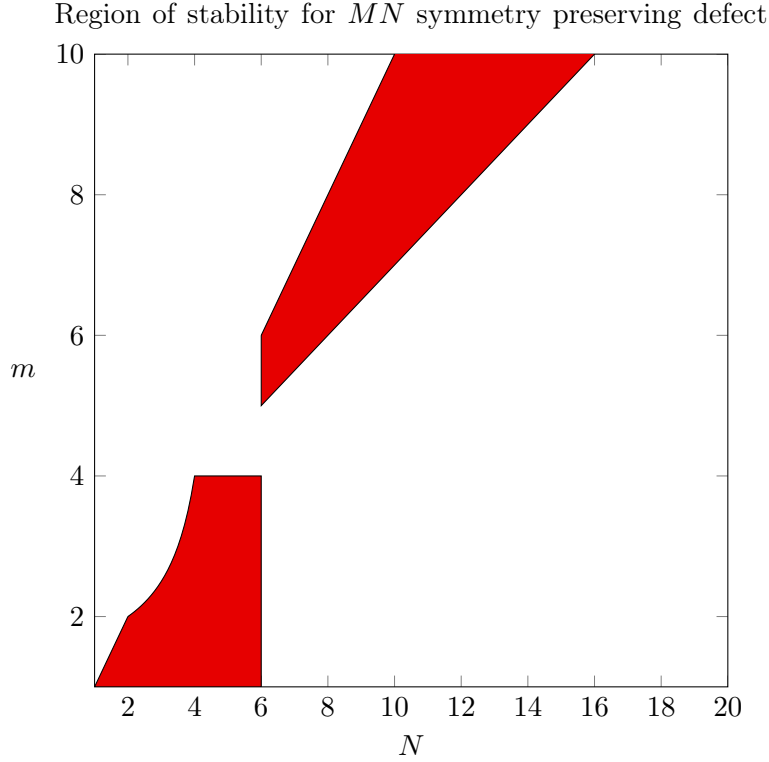
As is to be expected, setting $m=1$ yields the simple condition that $N\leq 6$ for stability, consistent with our analysis of the hypercubic model.

\subsection{Hypertetrahedral model}
Finally, we consider a bulk invariant under the symmetry group of a hypertetrahedron $T_N=S_{N+1}\times \mathbb{Z}_2$. As noted before, there are two bulk fixed points,
\begin{equation}
\begin{split}
    \lambda^+_{ijkl}=& \frac{\varepsilon}{24-15N+3N^2}\left(\delta_{ij}\delta_{kl}+\delta_{ik}\delta_{kl}+\delta_{il}\delta_{jk}\right)+\frac{(N-4)(N+1)\varepsilon}{3(N^2-5N+8)}\sum_\alpha e_i^\alpha e_j^\alpha e_k^\alpha e_l^\alpha\,,\\
    \lambda^-_{ijkl}=& \frac{\varepsilon}{3(N+3)}\left(\delta_{ij}\delta_{kl}+\delta_{ik}\delta_{kl}+\delta_{il}\delta_{jk}\right)+\frac{(N+1)\varepsilon}{3(N+3)}\sum_\alpha e_i^\alpha e_j^\alpha e_k^\alpha e_l^\alpha\,.
\end{split}
\end{equation}
The form of $(h,h)_\lambda$ is significantly more complicated than for the previously considered models due to the inclusion of the $e^\alpha$ vectors, and we will consequently not be able to derive full analytic results. The explicit expressions for $(h,h)_\lambda$ for the two fixed points are
\begin{equation}
\begin{split}
    (h,h)_{\lambda^+}=&\frac{(3N^2-15N+22)\varepsilon}{3(8-5N+N^2)}h_{ij}h_{ij}-\frac{\varepsilon}{3(8-5N+N^2)}(h_{ii})^2-\frac{(N^2-3N-4)\varepsilon}{3(8-5N+N^2)}\sum_\alpha (e^\alpha_i e^\alpha_j h_{ij})^2\,,\\
    (h,h)_{\lambda^-}=&\frac{(3N+7)\varepsilon}{3(N+3)}h_{ij}h_{ij}-\frac{\varepsilon}{3(N+3)}(h_{ii})^2-\frac{(N+1)\varepsilon}{3(N+3)}\sum_\alpha (e^\alpha_i e^\alpha_j h_{ij})^2\,.
\end{split}
\end{equation}
One can check numerically that these expressions will be positive definite for $N\leq 6$, and we expect that positive-definiteness will hold for all $N$. As $N$ increases, most elements of the vectors $e^\alpha$ will be suppressed like $1/N$, so that as the coefficient of the $h_{ij}h_{ij}$ term will also be asymptotically larger than the third coefficient. One should then always be able to complete the square to make this expression manifestly positive. As one can easily check, $d_{2,ij}=0$ in both cases, so there will always exist a symmetry preserving fixed point, which is
\begin{equation}
    h_{ij}=\frac{2(11-6N+N^2)}{3(8-5N+N^2)}\varepsilon \delta_{ij}\,,
\end{equation}
when we choose the bulk to be $T^+$ and
\begin{equation}
    h_{ij}=\frac{N+7}{3(N+3)}\varepsilon \delta_{ij}\,,
\end{equation}
when choosing a $T^-$ bulk. A simple check of stability for these points is examining the trace of the stability matrix, which one finds to be
\begin{equation}
    S_{ij,ij}=\begin{cases} \frac{N(24+5N-6N^2+N^3)}{6(8-5N+N^2)}\varepsilon & \lambda_{ijkl}=\lambda^+_{ijkl} \\
      \frac{N(9-N)(N+1)}{6(N+3)}\varepsilon & \lambda_{ijkl}=\lambda^-_{ijkl}
    \end{cases}\,.
\end{equation}
One can easily see that while the former of these is always positive, the latter is negative for $N\geq 10$. Given that we expect $(h,h)_\lambda$ to be positive definite in both of these cases for all $N$, the uniqueness theorem thus allows us to conclude that there exist no totally stable fixed points in a $T_N^-$ bulk for $N\geq 10$.

\section{Stability in scalar interface theories}\label{interfacesec}
Let us finally turn to the case of a $p=3-\varepsilon$ dimensional interface defect. The deformation is now cubic in the bulk field $\phi_i$, with the defect action taking the form
\begin{equation}
    S_{interface}=\int d^{d-1}\vec{x}\, h_{ijk}\phi_i(0,\vec{x})\phi_j(0,\vec{x})\phi_k(0,\vec{x})\,,
\end{equation}
for a totally symmetric interaction tensor $h_{ijk}$. The beta function for these defect couplings, calculated previously in \cite{Harribey:2023xyv,Harribey:2024gjn}, is found to be
\begin{equation}
    \beta_{ijk}=-\frac{\varepsilon}{2}h_{ijk}-\frac{1}{4}h_{iab}h_{jbc}h_{kac}+\big(\lambda_{ijab}h_{abk}+\text{Perms.}\big)\,.
    \label{eq:betainterface}
\end{equation}
Once more we write this beta function in terms of a metric and an $A$-function as $\beta^{I}=T^{IJ}\partial_J A$ for the choice
\begin{equation}
    A=-\frac{\varepsilon}{4}h_{ijk}h_{ijk}\,,\quad T_{ijk,abc}=\frac{1}{6}\left(\delta_{ia}\delta_{jb}\delta_{kc}+\text{Perms.}\right)+\frac{1}{2\varepsilon}h_{iam}h_{jbm}\delta_{kc}-\frac{2}{\varepsilon}(\lambda_{iab}\delta_{kc}+\lambda_{ikab}\delta_{jc}+\lambda_{jkab}\delta_{ic})\,.
\end{equation}
To analyze the stability properties of this beta function, let us first define the radius in coupling space to be the invariant $r^2=h_{ijk}h_{ijk}$. This $A$ function will evolve along RG flows according to the equation
\begin{equation}
    \frac{dA}{d\ln{\mu}}=r^2\left(\frac{\varepsilon^2}{4}-\frac{3\varepsilon}{2}\lambda_{ijab}\hat{h}_{abk}\hat{h}_{ijk}+\frac{\varepsilon}{8}\hat{h}_{iab}\hat{h}_{jbc}\hat{h}_{kac}\hat{h}_{ijk}r^2\right)\,,
\end{equation}
where $\hat{h}_{ijk}=h_{ijk}/r$ is a unit vector. All non-trivial fixed points must then lie on the surfaces defined by
\begin{equation}
    r^2=-\frac{2\varepsilon-12\lambda_{ijab}\hat{h}_{abk}\hat{h}_{ijk}}{\hat{h}_{iab}\hat{h}_{jbc}\hat{h}_{kac}\hat{h}_{ijk}}\,.
    \label{eq:rsolinterface}
\end{equation}
This expression will only have a real root as long as the right hand side is positive, so that either the numerator or denominator must be negative at each fixed point. One sees that close to the origin the sign of $dA/d\ln{\mu}$ will be determined by the sign of
\begin{equation}
    \frac{\varepsilon^2}{4}-\frac{3\varepsilon}{2}\lambda_{ijab}\hat{h}_{abk}\hat{h}_{ijk}\,.
\end{equation}
If $h^*_{ijk}$ is a fixed point and this expression is negative for $\hat{h}_{ijk}=h^*_{ijk}/\sqrt{h^*_{ijk}h^*_{ijk}}$, then perturbating towards the trivial defect will trigger an RG flow away from the fixed point. Stable, non-trivial fixed points must thus satisfy
\begin{equation}
    \frac{3\varepsilon}{2}\lambda_{ijab}\hat{h}_{abk}\hat{h}_{ijk}<\frac{\varepsilon^2}{4}\,,\qquad \hat{h}_{iab}\hat{h}_{jbc}\hat{h}_{kac}\hat{h}_{ijk}<0\,,
    \label{eq:interfacestabcondition}
\end{equation}
where the second equation follows from the unitarity requirement that $r$ is real. This analysis can be made more precise by contracting the stability matrix $S_{ijk,abc}$ with two factors of $h^*_{ijk}$ at a generic fixed point,
\begin{equation}
    h^*_{ijk}S_{ijk,abc}h^*_{abc}=-\frac{\varepsilon}{2}h^*_{ijk}h^*_{ijk}-\frac{3}{4}h^*_{iab}h^*_{jbc}h^*_{kac}h^*_{ijk}+3\lambda_{ijab}h^*_{abk}h^*_{ijk}=(\varepsilon-6\lambda_{ijab}\hat{h}_{abk}\hat{h}_{ijk})r^2\,,
\end{equation}
where in the last line we have used the fact that $\beta^*_{ijk}=0$ at the fixed point. If $\varepsilon<6\lambda_{ijab}\hat{h}_{abk}\hat{h}_{ijk}$ then $h^T\textbf{S}h<0$, which guarantees the existence of a relevant operator at the fixed point.

\begin{theorem}
    For $N\geq6$ there will be no totally stable fixed points if the bulk is taken to be critical. This can be expanded to any $N$ if the bulk is free.
\end{theorem}
\begin{proof}
We begin by noting that the stability matrix takes the form
\begin{equation}
\begin{split}
    S_{ijk,abc}=-\frac{\varepsilon}{12}(\delta_{ia}\delta_{jb}\delta_{kc}+\text{Perms.})+\frac{1}{3}(\lambda_{ijab}\delta_{kc}+\text{Perms.})-\frac{1}{24}(h_{iax}h_{jbx}\delta_{kc}+\text{Perms.})\,,
\end{split}
\end{equation}
where here Perms. refers to summing over cyclic permutations of both $i,j,k$ and $a,b,c$. The trace of the stability matrix will then be given by
\begin{equation}
    S_{ijk,ijk}=-\frac{N\lsp\varepsilon}{12}(2+3N+N^2)+(2+N)\lambda_{iijj}-\frac{N+2}{8}(h_{ijk}h_{ijk}+h_{ijj}h_{ikk})\,.
\end{equation}
Crucially, the last two terms in this expression involve only the sum of squares, and are thus always negative, so that to determine the sign of the trace we must examine the sign of the first two terms
\begin{equation}
    -\frac{N\lsp\varepsilon}{12}(2+3N+N^2)+(2+N)\lambda_{iijj}\,.
\end{equation}
For a free bulk, $\lambda_{ijkl}=0$, so that this will be negative for all $N$. We thus see that in a free bulk every fixed point will have a stability matrix with a negative trace, so that there can be no stable fixed points in a free bulk for any value of $N$. If we take the bulk to lie at an arbitrary critical point, we note that the value of $\lambda_{iijj}$ appearing in this expression can be bounded by (2.23) in \cite{Osborn:2020cnf}, so that
\begin{equation}
    -\frac{N\lsp\varepsilon}{12}(2+3N+N^2)+(2+N)\lambda_{iijj}<\frac{N \lsp\varepsilon}{12(N+8)}(32+22N+N^2-N^3)\,.
\end{equation}
Taking the number of scalars $N$ to be non-negative, the polynomial $32+22N+N^2-N^3$ is positive for $N\leq (3+\sqrt{73})/2\approx5.8$ and then negative for larger $N$. Restricting to an integer number of fields, we thus see that the trace of the stress tensor when the bulk is critical will necessarily be negative for $N\geq6$, so that the theorem holds.
\end{proof}

It is important to note that in this proof we have used the full stability matrix, which unavoidably includes even symmetry breaking deformations. As noted in \cite{Harribey:2024gjn} there do exist examples of fixed points which are stable stable within a one-dimensional symmetry preserving submanifold. The above theorem guarantees only that there exists at least one other direction in which these points are unstable.

To determine what happens at lower values of $N$, let us consider the cases separately. For a single scalar field the beta function has only one component given by
\begin{equation}
    \beta=-\frac{\varepsilon}{2}h-\frac{1}{4}h^3+3\lambda h\,,
\end{equation}
which has two fixed points\footnote{The bulk will have a $\mathbb{Z}_2$ symmetry $\phi\rightarrow -\phi$ for any value of $\lambda$, so that the two non-trivial roots of $\beta$ give physically equivalent fixed points.}
\begin{equation}
    h=0\,,\qquad h=\pm\sqrt{2}\sqrt{6\lambda -\varepsilon}\,.
\end{equation}
One sees that the non-trivial point will only correspond to a unitary dCFT if $6\lambda>\varepsilon$. As the stability matrix will be $(6\lambda-\varepsilon)/2$ for the trivial defect, and $\varepsilon-6\lambda$ at the interacting fixed point, one sees that there cannot exist a non-trivial, unitary stable fixed point for $N=1$. For two scalar fields there are four independent couplings, which in the notation of \cite{Harribey:2024gjn} are
\begin{equation}
    k_1=h_{111} \, ,\quad k_2=h_{222} \, , \quad g_1=h_{122} \, , \quad g_2=h_{112} \, .
\end{equation}
Examining the polynomial
\begin{equation}
    h_{iab}h_{jbc}h_{kac}h_{ijk}=4 g_1^3 k_1+6 g_1^2 k_2^2+12 g_2 g_1^2 k_2+12 g_2^2 g_1 k_1+6 g_2^2 k_1^2+4 g_2^3 k_2+3 g_1^4+12 g_2^2 g_1^2+3 g_2^4+k_1^4+k_2^4\,,
\end{equation}
one finds that it is only zero along the plane defined by $k_1=-g_1$ and $k_2=-g_2$. Examining the Hessian of $h_{iab}h_{jbc}h_{kac}h_{ijk}$ on this plane, one finds that this is a plane of minima, so that $h_{iab}h_{jbc}h_{kac}h_{ijk}$ is always non-negative. As $\hat{h}_{iab}\hat{h}_{jbc}\hat{h}_{kac}\hat{h}_{ijk}$ is just the restriction of this polynomial to the unit sphere, one sees that (\ref{eq:interfacestabcondition}) cannot be satisfied for any non-trivial fixed point. Thus, the only fixed point for $N=2$ that can possibly be stable is the trivial interface, $h_{ijk}=0$, whose stability depends upon the size of $\lambda_{ijkl}$.

For $N=3$, 4 and 5 we rely upon the numerical results presented in \cite{Harribey:2024gjn}. They claim to have a complete classification of interface fixed points for $N=3$, finding no totally stable fixed points beyond the trivial defect in any of the possible bulk models. For $N=4$ and 5 they find no non-trivial stable points for $O(N)$, hypercubic and hypertetrahedral fixed bulks, but do not study interfaces within all possible critical bulk models. On the weight of this numerical evidence, we conjecture that the conclusions at lower and higher values of $N$ will continue to be true here, and that in general the stable fixed point will be unique, and, if it exists, can only be the trivial defect point $h_{ijk}=0$.

The free defect will be stable when all cubic operators in the bulk theory have dimension greater than $3-\varepsilon$, and thus will be irrelevant deformations. The stability matrix at the free defect will only include contributions from the linear terms in the beta function, and takes the simpler form
\begin{equation}
    S_{ijk,abc}=-\frac{\varepsilon}{12}(\delta_{ia}\delta_{jb}\delta_{kc}+\text{Perms.})+\frac{1}{3}(\lambda_{ijab}\delta_{kc}+\text{Perms.})\,.
\end{equation}
For $N\leq 5$ there are a total of 19 fully interacting fixed points one may choose for the bulk critical model\cite{Osborn:2020cnf}, along with a number of decoupled fixed points. Explicitly checking the positive-definiteness of the above expression for each of these bulk models, we find that the only bulks for which the free point is stable are the $O(2)$, $O(3)$, $O(4)$ and $B_4$ models.

\section{Conclusion}
We have provided a general treatment of defect stability in multiscalar models, demonstrating that Michel's theorem survives for only for surface defects. Only a local version, still determined by the minimization of an $A$ function provided by the gradiency of the beta function, holds for line defects. For interface defects we instead find that unusually there exist no examples of non-trivial fixed points which are totally stable regardless of the form of the bulk theory. 

The local test of stability presented in this paper for line defects seems to be applicable to a wide range of systems. The observation that fixed points are embedded within a surface along which the beta function has no radial component relied on the construction of a metric $T_{ij}$, such that $\beta_i\propto T_{ij}\partial_j r^2$. Such a metric can be constructed very simply for any theory, simply by taking terms in the beta function and stripping away one factor of the coupling. $A\propto r^2$, suitably defined for theories with higher rank interaction tensors, should thus enjoy the same nice monotonicity properties that were seen in the case of line defects, with fixed points living along surfaces in coupling space defined by $d r/d\ln{\mu}=0$. The existence of this manifold may be a new, useful tool in investigating the structure of fixed points arising in various systems. For instance, it may be possible to apply the method of finding fixed points by extremizing the radius of this surface to prove that all possible biconical fixed points exist in multiscalar models as numerical evidence in \cite{Osborn:2020cnf} suggests. It is already known that there can exist multiple stable fixed points for scalar-fermion theories\cite{Pannell:2023tzc}, and while a modified form of Michel's theorem applies to purely scalar-deformations of these theories, it seems that a local, rather than global, test is necessary for understanding stability of these points with respect to generic deformations. 

One could also consider extending the bulk action to include fermions, with the bulk taken to lie at a critical point with a non-zero Yukawa interaction\cite{Pannell:2023tzc,Jack:2023zjt}, and asking whether or not the conclusions about fixed point stability reached in this work continue to hold. Line defects inside scalar-fermion bulk have previously been studied in \cite{Giombi:2022vnz,Pannell:2023pwz}, with the beta function being known to two-loops\footnote{The obstruction to gradiency for line defects noted in \cite{CarrenoBolla:2023vrv} at two-loops can be removed by the inclusion of a non-trivial metric $G^{IJ}$.}. An extension of surface and interface defects to include fermions in the bulk is also of interest.

\ack{We would like to thank Andreas Stergiou for various discussions about these ideas, as well as for providing comments on this manuscript. We would also like to thank Hugh Osborn for kindly looking over a draft of this work.}

\bibliography{main}

\end{document}